\documentclass{report}

\usepackage{amsmath, amsthm, amssymb, amsfonts}
\usepackage{hhline}
\usepackage{mathrsfs}
\usepackage{relsize}
\usepackage{float}
\usepackage{appendix}

\begin{document}

\title{Software Safety Demonstration and Indemnification}
\author{
    \Large{Odell Hegna}\\
    \MakeLowercase{\texttt{hegnaoc2@gmail.com}} \\
    contact: my beloved wife Carolyn D. Parsons, M.A., Psy.D.\\
    \MakeLowercase{\texttt{cdparsons@att.net}}
}

\theoremstyle{plain}
\newtheorem{theorem}{Theorem}[section]
\newtheorem{lemma}[theorem]{Lemma}
\newtheorem{corollary}[theorem]{Corollary}
\newtheorem{algorithm}[theorem]{Algorithm}
\newtheorem{conjecture}[theorem]{Conjecture}
\newtheorem{fact}[theorem]{Fact}
\newtheorem{axiom}[theorem]{Axiom}
\newtheorem{notation}[theorem]{Notation}

\theoremstyle{definition}
\newtheorem{definition}[theorem]{Definition}

\theoremstyle{remark}
\newtheorem*{remark}{Remark}
\newtheorem*{nomenclature}{Nomenclature}
\newtheorem*{example}{Example}
\newtheorem*{rationale}{Rationale}

\newcommand{\realnumbers}{\mathbb{R}}
\newcommand{\positives}{\mathbb{R}^+}
\newcommand{\nonnegatives}{\mathbb{R}_0^+}
\newcommand{\naturalnumbers}{\mathbb{N}}
\newcommand{\positivereals}{\mathbb{R}^+}
\newcommand{\positiveintegers}{\mathbb{Z}^+}
\newcommand{\naturalnumberz}{\mathbb{N}_0}
                              
\newcommand{\Inf}[1]{\mathrm{inf}(#1)}                  
\newcommand{\Sup}[1]{\mathrm{sup}(#1)}                  
\newcommand{\SetBuild}[2]{\{ {#1} \colon {#2} \} }      
\newcommand{\Idiom}[2]{\{ {#1}_{#2} \} }                
\newcommand{\IdiomPrime}[2]{\{ {#1}'_{#2} \} }         
\newcommand{\Single}[1]{{\{ {#1} \}}}                   
\newcommand{\domain}[1]{{\operatorname{dom}{#1}}}
\newcommand{\range}[1]{{\operatorname{ran}{#1}}}
\newcommand{\codomain}[1]{{\operatorname{codom}{#1}}}
\newcommand{\nullplex}{\lbrace\varnothing\rbrace}
\newcommand{\powerset}[1]{\mathscr{P}({#1})}

\def\sync{\mathop{\textstyle{\mathrm{sync}}}}
\def\test{\mathop{\mathbf{t}}}
\newcommand{\Instance}[1]{{\mathsf{\footnotesize{stochastic}\thinspace(}{#1}\mathsf{)}}}
\newcommand{\OP}{\mathcal{O} \negthickspace \mathcal{P}}          

\newcommand{\Fr}[1]{\mathfrak{#1}}                      
\newcommand{\Cp}[1]{\Pi\mathsf{#1}}                     
\newcommand{\CP}[1]{\prod{#1}}                          
\newcommand{\Ip}[1]{\mathfrak{#1}}                      
\newcommand{\Auto}[1]{\mathfrak{#1}}                    
\newcommand{\AutoInv}[1]{\mathfrak{#1}^{-1}}            
\newcommand{\Card}[1]{\lvert{#1}\rvert}                 
\newcommand{\Norm}[1]{\Vert{#1}\Vert}                   
\newcommand{\Fty}[1]{\mathscr{#1}}                      
\newcommand{\Ftylc}[1]{\mathit{#1}}                     
\newcommand{\Act}[1]{\mathsf{#1}}			
\newcommand{\Frm}[1]{\mathbf{#1}}                       
\newcommand{\Prc}[1]{\mathsf{#1}}			
\newcommand{\Stp}[1]{\mathbb{#1}}                       
\newcommand{\Stplc}[1]{\mathit{#1}}                     
\newcommand{\Con}[1]{\mathcal{#1}}                      
\newcommand{\Closure}[1]{\widehat{#1}}                
\newcommand{\Fld}[1]{\mathscr{#1}}                      
\newcommand{\Basis}[2]{\Closure{#1} \times \Closure{#2}}      
\newcommand{\Abstract}[1]{\mathit{abstraction}({#1})}   
\newcommand{\Dquo}[1]{\textquotedblleft{}#1\textquotedblright}

\newcommand{\Acro}[1]{{\small{#1}}}
\newcommand{\MilStd}[0]{{\small{MIL-STD-882}}}          
\newcommand{\MilStdE}[0]{{\small{MIL-STD-882E}}}	

\newcommand{\ConfigSpace}[0]{\textstyle{\Closure{\Phi} \times \Closure{\Xi} \times \Act{A} \times \Fty{F}\! \times \Lambda}}
\newcommand{\ConfigStd}[0]{(\phi, \xi, \Act{a}, \Ftylc{f}, \lambda)}
\newcommand{\ConfigSub}[1]{(\phi_{#1}, \xi_{#1}, \Act{a}_{#1}, \Ftylc{f}_{#1}, \lambda_{#1})}
\newcommand{\ConfigScript}[1]{\mathfrak{#1}}
\newcommand{\Prev}[0]{\mathrm{prev}}
\newcommand{\Curr}[0]{\mathrm{curr}}
\newcommand{\Next}[0]{\mathrm{next}}

\newcommand{\indexfont}[1]{\mathscr{#1}}                
\newcommand{\Restrict}[2]{{#1}\negmedspace\restriction_{#2}\thinspace}
\newcommand{\xth}[1]{{#1}^{\text{th}}}

\newcommand{\edge}[1]{{\operatorname{edge}{#1}}}


\def\E{\mathop{\textstyle{\mathsf{E}}}}
\newcommand{\Ex}[1]{\mathop{\textstyle{\mathsf{E}_{#1}}}}
\def\Var{\mathop{\textstyle{\mathsf{Var}}}}

\newcommand{\Bigast}[1]{{\operatorname{\mathlarger{\mathlarger{\ast}}}\!{#1}}}
\newcommand{\absc}[1]{{\operatorname{absc}{#1}}}	
\newcommand{\ord}[1]{{\operatorname{ord}{#1}}}		
\newcommand{\Blocks}[1]{{\operatorname{\boxplus}{(#1)}}}
\newcommand{\blocks}{{\operatorname{\boxplus}}}

\newcommand{\tranbasic}[1]
{
    {#1} \circ {\operatorname{\boxplus}}_{\domain{#1}}
}
\newcommand{\transducer }[1]		
{
    ({#1} \circ {\operatorname{\boxplus}}_{\domain{#1}})
}
\newcommand{\transduction }[2]		
{
    ({#1} \circ {\operatorname{\boxplus}}_{\domain{#1}})(#2)
}
\newcommand{\tranfunc}[2]		
{
    \operatorname{absc} {
        (({#1} \circ {\operatorname{\boxplus}}_{\domain{#1}})(#2))
    }
}
\newcommand{\tranaddr}[2]		
{
    \operatorname{ord} {
        (({#1} \circ {\operatorname{\boxplus}}_{\domain{#1}})(#2))
    }
}
\newcommand{\tranresp}[2]		
{
    (\operatorname{absc} {
        ({#1} \circ {\operatorname{\boxplus}}_{\domain{#1}})(#2))(#2)
    }
}
\newcommand{\tranframe}[2]		
{
    ({#2}, 
    (\operatorname{absc} {
        (({#1} \circ {\operatorname{\boxplus}}_{\domain{#1}})(#2))(#2)
    }
    ))
}


\setlength{\parindent}{0in}
\setlength{\parskip}{0.1in}

\begin{abstract}
Computers may control safety-critical operations in machines having embedded software. This memoir proposes a regimen to verify such algorithms at prescribed 
levels of statistical confidence. 

The United States Department of Defense standard for system safety engineering (\Acro{MIL-STD-882E}) defines development procedures for safety-critical systems. 
However, a problem exists: the Standard fails to distinguish quantitative \emph{product} assurance technique from categorical \emph{process} assurance method for software development. Resulting is conflict in the technical definition of the term \emph{risk}. 

The primary goal here is to show that a quantitative risk-based product assurance method exists and is consistent with hardware practice. Discussion appears in two major parts: theory, which shows the relationship between automata and software; and application, which covers demonstration and indemnification. Demonstration is a technique for generating random tests;  indemnification converts pass/fail test results to compound Poisson parameters (severity and intensity). Together, demonstration and indemnification yield statistical confidence that safety-critical code meets design intent. Statistical confidence is the keystone of quantitative product assurance. 

A secondary goal is resolving the conflict over the term \emph{risk}. The first meaning is an 
accident model known in mathematics as the compound Poisson stochastic process, and so is 
called \emph{statistical} risk. Various of its versions underlie the theories of safety and 
reliability. The second is called \emph{developmental} risk. It considers software autonomy, 
which considers time until manual recovery of control. Once these meanings are separated, 
\Acro{MIL-STD-882} can properly support either formal quantitative safety assurance or empirical process robustness, which differ in impact.


\end{abstract}
\maketitle


\chapter{Prologue} 

\section{About this memoir}\label{S:ABOUT} 

\subsection{Copyright}\label{S:COPYRIGHT} 
This document may be freely copied or modified in accordance with the Creative Commons 
Attribution license\footnote{http://creativecommons.org/licenses/by/3.0/}. 

\subsection{Preface to the fourth revision}\label{S:PREFACE} 
This fourth revision was to improve the document's exposition as applied mathematics. 
This goal has met limited success; it is challenged by the author's meager education 
and his failing health. 

The author deeply regrets that his health now dictates that he is unable to complete 
this revision. Consequently, this incomplete edition is classified as a memoir. He bids 
farewell to this portrayal of reactive systems, and wishes all interested parties well in 
continued advancement. His fondest wish is that this memoir will prove useful to researchers. 

.
\subsection{Acknowledgment} \label{S:PERSONAL_DISCOVERY} 
The author acknowledges W. Ross Ashby and his 1956 pioneering work in cybernetics\cite{wA56}. 
His depiction of the transducer [The Determinate Machine and The Machine With Input, 
Chapters 3 and 4] was the author's personal inspiration for the reactive (semi-deterministic) 
actuated automaton. Here this topic is covered in chapter \ref{S:RAA_CHAPTER}.

\subsection{Credentials}\label{S:CREDENTIALS} 
The author holds an undergraduate degree in mathematics from the University of Minnesota. He is a retired safety and reliability engineer with more than three decades of software experience. 

\subsection{Advocacy}\label{S:ADVOCACY} 
The author is a critic of \MilStd, the United States Department of Defense standard for systems 
safety. A fault occurring in the Standard is absence of quantitative assurance for software. 
He objects to the prescribed software safety metrics, which result in incomparable measures 
of risk between software and hardware hazards. The author's credentials mark this memoir as 
experienced but non-scholarly advocacy. 

\subsection{Approach}\label{S:APPROACH} 
This memoir collects related career experiences of the author into a mathematical 
discussion. 

A custom automaton lays a foundation for software quantitative assurance. It 
accomplishes this by examining the significance of a random sample of bounded software 
trajectories passing through a common point. Such a collection, when tested, is here 
called a safety demonstration. No error is tolerated in a safety demonstration. 

Indemnification converts the test result to Poisson probability, which is consistent with 
hardware practice. The test result is expressed as statistical confidence.

\section{Management summary}\label{S:EXECUTIVE_SUMMARY} 
United States {\MilStdE}, a widely recognized system safety standard, mistakenly prescribes two incompatible safety 
analysis methodologies, one for hardware and the other for software. The safety method for software was formulated by 
focusing on the \Dquo{shift-left}\cite{wW17shiftleft} debugging strategy of development engineers. When this partial 
but incomplete strategy is imposed on safety engineers, quantitative assurance of software safety is lost. The present 
Standard ignores an important objective of system safety engineering: providing a {\footnotesize{\emph{QUANTITATIVE}}} 
level of assured safety, having common measure between hardware and software. Indeed this error is so 
fundamental that one must question whether systems safety was properly represented during this standard's review. 
Properly managed, the system safety discipline requires common philosophy and measure of safety in order to understand 
and rank heterogeneous hazards. 

\MilStd{} indisputably does not provide a common measure of safety assurance. The present Standard is not wrong 
as it stands \emph{for development engineers}. However, the Standard fails one of the leading expectations 
of safety engineering, namely to express the risks of hazards in common units. Thus \MilStd{} suffers a 
fundamental error which renders it unacceptably wrong. The prescribed safety methods for hardware and 
software diverge in both metrics and ideology. Harm ensues because modern engineered systems are usually neither 
pure hardware nor pure software; thus they require common measure to render hazard threats comparable. 

Development engineers characteristically have superior knowledge of project detail, whereas safety engineers 
appreciate the quantitative risks underlying hazard identification and analysis. Hazard analysis is a separate 
engineering topic.\footnote{A comprehensive work on the subject is  
\emph{Engineering a Safer World}\cite{nL11} by Leveson}

The overall process is improved by fine-tuning the duties of these personnel to their particular abilities: development 
and debugging for software engineers and quantitative assurance for safety engineers. Currently \MilStd{} does not 
differentiate the duties of these two personnel groups, but recommends a team approach with oversight responsibility 
for safety engineers.

Bug discovery is a necessary prelude to assurance, which is not direct verification of logic itself, but verification of a 
random sample of software trajectories representing the logic under examination. This approach is necessary because gross logic 
is assembled from smaller combinatory pieces, and whether these pieces are always equivalent to the whole must be quantitatively 
assured. This is also the reason that assurance is a statistical task with levels of confidence depending on sample size. 

The author recognizes shift-left tactics as important to efficient development; no change to process for development engineers is 
proposed here. However, we do wish to add software assurance procedure for safety engineers. In this sense, \emph{assured} is not 
merely synonymous with feeling good about the development process. Assurance tests are specially structured to affirm safety at 
known statistical confidence levels. These procedures are plainly missing from \MilStd; it should be amended to include metrics 
supporting quantitative assurance of safety hazards. 

\section{Apologies to the reader}\label{S:APOLOGY} 
The author apologizes that the concepts discussed here, being distilled from his 
personal experience, are new to him but not to all. 

This memoir falls short of academic standards of quality; it benefits from neither 
literature search nor peer review. The experienced reader may find terms in nonstandard 
context. The author has strived to maintain consistency, but admits deficiency in 
standardization of nomenclature, a consequence of writing in isolation. The author 
apologizes for resulting inconvenience.

There are worthy readers who would prefer a traditional engineering approach 
(examples) to this topic; however mathematics is central in what 
follows. A mathematical foundation is necessary, and this memoir is a step 
in that direction. The author was emphatically not a mathematician; his 
education is over-matched even to begin this work. The author encourages 
interested academically qualified individuals to advance this memoir into 
proper research.

\chapter{Rudiments of discrete reactive systems} \label{S:RUDIMINARY_REACTIVE_SYSTEMS_CHAPTER} 
This chapter discusses core structures of discrete systems theory. Definitions of very basic 
concepts such as \emph{ensemble}, \emph{class}, and their elementary operations appear in 
Groundwork, Appendix \ref{S:ENSEMBLE} ff. 

\section{Description of discrete systems}
The operation of discrete systems is composed of chains woven from paired units of stimulus 
and response. Each stimulus appears in either of two types, \emph{deterministic} or 
\emph{stochastic}, and possesses a \emph{value}. Deterministic 
stimuli are computable and retain their values until re-assigned, a feature that is enabled 
by buffered storage. Stochastic stimuli are volatile, requiring observation rather than 
computation, and their value is defined  only at the instant of observation. This value 
may be copied to deterministic storage. Example stochastic stimuli are inputs from sensors and 
remote commands to robotic apparatus. Systems containing both stimulus types are called 
\emph{reactive}.

For any combination of stimuli, discrete reactive systems construct a unique \emph{response}. 
The agent transforming stimuli into response is called a \emph{functionality}.
Since responses are computable, they are recorded in deterministic storage. Therefore, as a mapping, a functionality has domain consisting of deterministic and stochastic stimuli, and 
codomain similarly composed of deterministic storage.

In its characteristic chain of stimulus and response, successive links are called frames. 
Each frame is composed of a stimulus (starting condition) and a response (ending condition). 
Frames are not independent because of a mechanism called feedback. Feedback expresses the 
principle that the current frame's response is included in the next frame's stimulus 

As a non-reactive example consider a clock having pendulum, cog-wheels and escapement. 
At each tick, accomplished gear train movement becomes input to the next. The ending 
condition generated in the current frame feeds forward into the starting condition of the 
next frame. 

\begin{axiom}[deterministic stimulus] \label{A:PERSISTENT_STIMULUS}
Class $\Closure{\Phi}$ represents deterministic stimulus, the value of which persists between assignments. 
\end{axiom}

\begin{axiom}[stochastic stimulus] \label{A:STOCHASTIC_STIMULUS}
Class $\Closure{\Xi}$ represents stochastic stimulus, the value of which is 
instantaneous and not predictable free of error. A stochastic stimulus 
is \Dquo{read-only}, used but set outside the system. 
\end{axiom}

\begin{axiom}[response] \label{A:RESPONSE}
Class $\Closure{\Phi}$ contains the system's unique response to the reactive stimulus. 
\end{axiom}

\begin{axiom}[reactive] \label{A:REACTIVE}
A stimulus and response is \Dquo{reactive} if it is a mapping from the union of 
deterministic and stochastic space into deterministic space, symbolically 
$\Closure{\Phi}\Closure{\Xi} \to \Closure{\Phi}$ (see \S\ref{S:DETERMINISTIC_STOCHASTIC_PRTN}). 
\end{axiom} 

\begin{definition}[reactive] \label{D:REACTIVE}
A stimulus and response is \Dquo{reactive} if it is a mapping from the union of 
deterministic and stochastic space into deterministic space, symbolically 
$\Closure{\Phi}\Closure{\Xi} \to \Closure{\Phi}$ (see \S\ref{S:DETERMINISTIC_STOCHASTIC_PRTN}). 
\end{definition} 

\begin{definition}[reactive basis] \label{D:REACTIVE_BASIS}
Two classes $\Closure{\Phi}$ and $\Closure{\Xi}$ are a \emph{reactive basis}, designated $\Basis{\Phi}{\Xi}$, 
if there is a mapping between stimulus and response $\Closure{\Phi}\Closure{\Xi} \to \Closure{\Phi}$, 
and if this mapping is disjoint ($\domain{\Phi} \cap \domain{\Xi} = \varnothing$).
\end{definition}

\section{Principles of stochastic stimulus}
Reactive stimulus in systems is a composite of the two stimulus types, deterministic and stochastic. 
Stochastic stimulus is required to complete the value of reactive stimulus. In mathematics stochastic stimuli are values which are not defined until used. In accordance with the aforementioned physics interpretation, we require a sequence of stochastic stimuli to be read in sequential order, at the instant corresponding to that order. 

\begin{notation} \label{N:PSEUDO_STOCHASTIC}
The pseudo-function $\xi = \Instance{\Closure{\Xi}}$ means that $\xi$ is a random sample of population $\Closure{\Xi}$.
\end{notation}

\begin{remark}
Distributions, moments, and autocorrelations may be important in real 
problems, but will not be used here. Randomness does not deny these properties.
\end{remark}

\subsection{Cybernetics}
We routinely picture time as a continuum of instants. In this framework, each event of discrete 
systems theory occupies an instant and occurs in a duration (interval). Bounded collections of 
events possess either a time or a starting and ending time. In discrete processes frames constitute a partition of time. Frames adjacent in time are connected. 

\begin{axiom} \label{A:EXIST_CONFIG}
Before observation, a stochastic stimulus is non-existent as a realized value. 
\end{axiom}
\begin{remark}
Before the current cycle (definition \ref{D:CYCLE}), the current stochastic stimulus is an object of probabilistic 
uncertainty (a random variable). In the current cycle it transitions into a realized observation. 
\end{remark}

%

The automaton resembles an automorphism, except certain of its arguments are unknown until their instant of use. In physics 
stochastic stimulus is not computable, and the only way to determine its exact value is through 
observation. 

Deterministic stimulus has the contrasting property that once its value is set, it retains that 
value until re-set. This is the paradigm associated with ordinary computer memory transactions.

\section{Frame space} 
The frame is a two-part structure consisting of starting and ending conditions. Interpreted in systems vocabulary, 
a frame's starting condition is a stimulus and its ending condition is a response. 

\begin{definition} \label{D:FRAME_SPACE}
The \emph{frame space} $\Frm{F}$ of reactive basis $\Basis{\Phi}{\Xi}$ is the set $\Closure{\Phi}\Closure{\Xi} \times \Closure{\Phi}$.
A member $\Frm{f} \in \Frm{F}$ is a \emph{frame}.
\end{definition}
\begin{nomenclature}
Let $\Frm{f} = (\phi \xi, \varphi) \in \Closure{\Phi}\Closure{\Xi} \times \Closure{\Phi}$ be a frame. The choice $\phi\xi \in \Closure{\Phi}\Closure{\Xi}$ is the 
frame's starting condition (abscissa) and $\varphi \in \Closure{\Phi}$ is the frame's ending condition (ordinate).
\end{nomenclature}

\begin{definition}\label{D:ABSCISSA_PROJECTION}
Let $\Basis{\Phi}{\Xi}$ be a basis with frame space $\Frm{F} = \Closure{\Phi}\Closure{\Xi} \times \Closure{\Phi}$. Define the \emph{abscissa}
projection $\absc{}: \Frm{F} \to \Closure{\Phi}\Closure{\Xi}$ by $(\phi\xi, \varphi) \stackrel{\absc{}}{\mapsto} \phi\xi$. Define the \emph{ordinate}
projection $\ord{}: \Frm{F} \to \Closure{\Phi}$ by $(\phi\xi, \varphi) \stackrel{\ord{}}{\mapsto} \varphi$.
\end{definition}

\begin{definition}\label{D:PERSISTENT_STOCHASTIC_COMPONENTS}
Let $\Basis{\Phi}{\Xi}$ be a basis with deterministic-stochastic partition $\Psi = \Phi\Xi$ (see \S\ref{S:DETERMINISTIC_STOCHASTIC_PRTN}). 
Suppose $\Frm{f}$ is a frame in $\Closure{\Phi}\Closure{\Xi} \times \Closure{\Phi}$. The \emph{reactive} stimulus of frame $\Frm{f}$ is 
$\psi = \phi\xi = \absc{\Frm{f}}$. The \emph{stochastic} excitation stimulus of frame $\Frm{f}$ is 
$\xi = \Restrict{(\absc{\Frm{f})}}{\domain{\Xi}}$. Similarly, the \emph{deterministic} stimulus of frame 
$\Frm{f}$ is $\phi = \Restrict{(\absc{\Frm{f})}}{\domain{\Phi}}$.
\end{definition}

\begin{nomenclature}
In the above, both $\phi$ and $\varphi$ are members of $\Closure{\Phi}$. 
Interpretation of the choice space $\Closure{\Phi}$ is contextual. 
It is a component of stimulus in context of a frame's starting condition, and is the system's response in the context 
of a frame's ending condition. 
\end{nomenclature}

\section{Feedback} \label{S:FEEDBACK} 
Two frames may be related such that the ending condition of the first frame is replicated as 
a subset of the second frame's starting condition. This stipulation is called feedback; it's 
conveniently expressed as a mapping restriction.

\begin{axiom} \label{D:FEEDBACK}
\emph{Feedback} is the principle that the current frame's response becomes the next frame's deterministic stimulus.
\end{axiom}
\begin{remark}
For consecutive frames $\Frm{f}$ and $\Frm{f}'$, this condition is 
written $\ord{\Frm{f}} = \Restrict{(\absc{\Frm{f}'})}{\domain{\Phi}}$. 
\end{remark}

\begin{definition}[feedback-coupled frames] \label{D:FEEDBACK_COUPLED}
Let $\Basis{\Phi}{\Xi}$ be a reactive basis with frames $\Frm{f}$ and $\Frm{f}'$ in frame 
space $\Frm{F} = \Closure{\Phi}\Closure{\Xi} \times \Closure{\Phi}$. Frame $\Frm{f}$ is (directionally) feedback-coupled to $\Frm{f}'$ 
if $\ord{\Frm{f}} = \Restrict{(\absc{\Frm{f}'})}{\domain{\Phi}}$.
\end{definition}

\begin{remark} 
Suppose the current frame is $(\phi\xi, \varphi)$ and the next frame is $(\phi'\xi', \varphi')$. 
The feedback equality $\phi' = \varphi$ is intrinsic to discrete reactive systems. 
\end{remark} 

\begin{definition}[coupling for sequences of frames] \label{D:SEQUENTIALLY_COUPLED}
Let $\Basis{\Phi}{\Xi}$ be a reactive basis with sequence of frames 
$\lbrace \Frm{f}_n \rbrace \colon \naturalnumbers \to \Closure{\Phi}\Closure{\Xi} \times \Closure{\Phi}$. The sequence is 
\emph{coupled by feedback} if $\Frm{f}_i$ is coupled to $\Frm{f}_{i+1}$ for each $i \geq 1$.
\end{definition}

\begin{definition}[process] \label{D:PROCESS}
With $\Basis{\Phi}{\Xi}$ a reactive basis, a \emph{process} is a feedback-coupled sequence
of frames $\naturalnumbers \to \Closure{\Phi}\Closure{\Xi} \times \Closure{\Phi}$. 
\end{definition}

\section{Functionality} \label{S:FUNCTIONALITY} 

The functionality is useful to portray the frame as a transformation from the reactive stimulus to the response space, 
which is identical to the deterministic space. 

The functionality determines how the frame goes from reactive stimulus to response. 
If $\Frm{f}_i = (\psi_i, \varphi_i) = (\phi_i\xi_i, \varphi_i)$ is the $i^{\text{th}}$ 
process frame, this concept permits writing $\varphi_i = {\Ftylc{f}}(\psi_i) = {\Ftylc{f}}(\phi_i\xi_i)$, where $\Ftylc{f}$ 
is a functionality:

\begin{definition} \label{D:FUNCTIONALITY}
Let $\Basis{\Phi}{\Xi}$ be a reactive basis. Any mapping $\Ftylc{f} \colon \Closure{\Phi}\Closure{\Xi} \to \Closure{\Phi}$ is a functionality
(that is, if $\Ftylc{f} \in {\Closure{\Phi}}^{\Closure{\Phi}\Closure{\Xi}}$).
\end{definition}

\begin{remark}[functionality versus function]
In its programming sense, the term \Dquo{function} will not be used here. A mathematical functionality differs from 
a software function; functionalities take a global approach to the function's argument list. By virtue of its ordered 
argument list, a programming function is effectively a class of functionalities.
\end{remark}

\begin{definition}[procedure] \label{D:PROCEDURE}
Let $\Fty{F} \subseteq {\Closure{\Phi}}^{\Closure{\Phi}\Closure{\Xi}}$ be a finite collection of functionalities. 
A \emph{procedure} is a sequence $\Idiom{\Ftylc{f}}{n} : \naturalnumbers \to \Fty{F}$. 
\end{definition}

\section{Consistency of frames and functionalities} \label{S:CONSISTENCY} 
The relation holding between frame $\Frm{f} \in \Closure{\Phi}\Closure{\Xi} \times \Closure{\Phi}$ and functionality
$\Ftylc{f} \colon \Closure{\Phi}\Closure{\Xi} \to \Closure{\Phi}$ is membership: either $\Frm{f} \in \Ftylc{f}$ or $\Frm{f} \notin \Ftylc{f}$.

\begin{definition} \label{D:BASIC_CONSISTENCY}
Let $\Frm{f}$ be a frame and $\Ftylc{f}$ be a functionality. The frame and the functionality are \emph{consistent} 
if $\Frm{f} \in \Ftylc{f}$ (that is, $\Frm{f} = (\psi, \varphi) = (\psi, \Ftylc{f}(\psi))$).
\end{definition}

\begin{notation}[anonymous sequence] \label{N:SEQUENCE_NOTATION}
A sequence in a set $S$ is some mapping $\sigma \colon \naturalnumbers \to S$ -- that is, $\sigma \in S^\naturalnumbers$. The
anonymous sequence convention allows reference to a sequence using the compound symbol $\lbrace s_n \rbrace$, understanding
$s \in S$. The symbol $s_i$ denotes that term $(i, s_i) \in \lbrace s_n \rbrace$. 
\end{notation}

\begin{remark}
The convention is clumsy when expressing functional notation; for instance $s_i = \lbrace s_n \rbrace(i)$ means 
$i \stackrel{\lbrace s_n\!\rbrace}{\mapsto} s_i$.
\end{remark}

\begin{definition} \label{D:SEQUENTIAL_CONSISTECY}
Let $\Idiom{\Frm{f}}{n}$ be a sequence of frames and $\Idiom{\Ftylc{f}}{n}$ be a sequence of functionalities. The 
sequences of frames and functionalities are \emph{consistent} if $\Frm{f}_i \in \Ftylc{f}_i$ for each $i \geq 1$ (that is,
$\Frm{f}_i = (\psi_i, \varphi_i) = (\psi_i, \Ftylc{f}_i(\psi_i))$).
\end{definition}

\chapter{Discrete categorical regulation} \label{S:REGULATORY_STRUCTURES_CHAPTER} 
Regulatory structures for the reactive actuated automaton are partitions, actuators, 
and related transducers. The domain sponsoring regulation is reactive space.
Transducer labels have a subordinate role.

\section{Lookup tables} \label{S:LOOKUP_TABLES} 
A lookup table $T$ is a finite computer science structure of paired indices and values. 
It is secondarily known as an associative array, and denoted that $(i, v) \in T$ or 
$v = T(i)$. As the notation indicates, the relation between indices and values must 
be a mapping.

Lookup tables are alternatively organized via their indices' level sets. 
Let $\ell$ be a level set of indices, so $ \ell = \lbrace i_1, i_2, ... i_k \rbrace$.
In canonical compound form, a lookup table has the property 
$(\ell, v) \in T$ and $(\ell', v) \in T$ implies $\ell = \ell'$. 

Lookup tables are easily implemented in most programming languages. 

\subsection{Classification of lookup tables} \label{S:LOOKUP_CLASSIFICATION} 
We identify two basic types of lookup table which are equivalent but differ in 
set-theoretic structure. The direct-indexed form is a transducer and the level-set 
type is an actuator.

\section{Partitions} \label{S:PARTITIONS} 
In the context of a partition, the level set of \S\ref{S:LOOKUP_TABLES} is known as a block. 

\begin{definition}[block]
A \emph{block} of a collection $X$ is a non-empty subset $B$, such that $B \neq \varnothing$ and $B \subseteq X$.
\end{definition}

\begin{definition}[finite partition]
A \emph{finite partition} of a collection $X$, denoted $\#(X)$, is a finite set of blocks $B_i$, $1 \leq i \leq n$, such that 
each $B_i \neq \varnothing$, $B_i \cap B_j = \varnothing$ if $i \neq j$, and $X = \bigcup B_i$.
\end{definition}

We have defined $\#$ as a single instance of a partition, rather than the set of all partitions.

\begin{definition}[rho function] \label{D:RHO_FUNCTION}
Let $X$ be a set with partition $\#(X)$. Suppose $B \in \#(X)$ is a block and $\chi \in X$ is a member.
The \emph{rho} function is $\varrho_B(\chi) = 
\left\{
    \begin{array}{ll}
        B           & \mbox{if } \chi \in B \\
        \varnothing & \mbox{if } \chi \notin B
    \end{array}
\right. .$ 
\end{definition}

\subsection{The containing block function and families of partitions} \label{S:PARTITION_FAMILY} 

The \emph{containing block} function converts a point of reactive space into its containing 
partition block.

\begin{definition}[containing block of a partition] \label{D:POINT_BLOCK_CONVERSION}

Suppose $\chi \in X$ is a member of set $X$, $\#(X)$ is a partition of $X$, and $B \in \#(X)$ 
is a block of the partition. The \emph{containing block} mapping $\blocks$,  is defined as 
$\SetBuild{(\chi, B_\chi)}{\chi \in X \text{ and } B_{\chi} = \bigcup_{B \in \#(X)} \varrho_B(\chi)}$ 
(see definition \ref{D:RHO_FUNCTION} for rho function).
\end{definition}

\begin{remark}
The containing block function has prototype $\blocks \colon X \to \#(X)$. 
\end{remark}

The abstract set $X$ is now realized as $\Closure{\Psi}$, the space of reactive stimuli. 
Suppose we have a finite collection $\mathcal{P}$ of finite partitions of reactive space 
(that is, for each $p \in \mathcal{P}$, $p = \#(\Closure{\Psi})$ for some partition $\#(\Closure{\Psi})$). 
This mapping is symbolized $\blocks$, and it is associated with some implicitly understood 
partition $\#(\Closure{\Psi})$. The notion of a solitary containing block may be extended into a 
family of mappings indexed by $p \in \mathcal{P}$.

\begin{notation}
Let $\mathcal{P}$ be a collection of partitions of reactive space.
For each $p \in \mathcal{P}$,  the corresponding containing block function is denoted 
$\blocks_p$.
\end{notation}

\section{Actuators} \label{S:ACTUATORS} 

%


In software engineering, we take the actuator as a uniform comb-structured conditional 
statement: if \emph{condition$_1$} then \emph{consequence$_1$}, elsif \emph{condition$_2$} 
then \emph{consequence$_2$}, \ldots{} else \emph{consequence$_n$}. Each condition is a set 
of reactive stimuli, and the consequences are homogeneous and possibly compound-valued.

\begin{definition}[actuator] \label{D:ACTUATOR}
An \emph{actuator} is a mapping $\Act{a} \colon \#(\Closure{\Psi}) \to N$ from the blocks of a 
finite partition of reactive space to an abstract space $N$. 
\end{definition}

\begin{remark}
The actuator of definition \ref{D:ACTUATOR} has the form 
$\{ B_1 \mapsto \nu_1, \thickspace\cdots\thickspace , B_n \mapsto \nu_n \}$, 
where each block $B_i \in \#(\Closure{\Psi})$ and $\nu_i \in N$.
\end{remark}

\begin{definition}[canonical form] \label{D:CANONICAL_FORM} 
An actuator $\Act{a} \colon \#(\Closure{\Psi}) \to N$ is in \emph{canonical form} if 
$(B, \nu) \in \Act{a}$ and $(B', \nu) \in \Act{a}$ implies ${B} = {B'}$ for any $\nu \in N$.
\end{definition}

\begin{lemma} \label{L:ACTUATOR_DOMAIN}
Any actuator $\Act{a}$ has domain $\domain{\Act{a}}$, which is itself a partition of the form $\#(\Closure{\Psi})$. 
\end{lemma}
\begin{proof}
Definition \ref{D:ACTUATOR} states that any actuator $\Act{a}$ is a mapping $\#(\Closure{\Psi}) \to N$. 
By equivocation, $\domain{\Act{a}} = \#(\Closure{\Psi})$ is a partition of $\Closure{\Psi}$.
\end{proof}

\section{Transducers} \label{S:TRANSDUCERS} 
\begin{lemma}[re-indexed actuator] \label{L:ACTUATOR_REINDEX}
Suppose $\Act{a} \colon \#(\Closure{\Psi}) \to N$ is an actuator per definition \ref{D:ACTUATOR}, 
and $\blocks_{\domain{\Act{a}}} \colon \Closure{\Phi} \to \#(\Closure{\Phi})$ is its containing block function 
(definition \ref{D:POINT_BLOCK_CONVERSION}).
An actuator $\Act{a}$ is re-indexed from partition blocks $\#(\Closure{\Psi})$ to reactive stimulus space 
$\Closure{\Psi}$ through the composition: 
\[
\tranbasic{\Act{a}} \colon \textstyle{\Closure{\Psi}} \to N.
\]
\end{lemma}

\begin{proof}
By definition \ref{D:ACTUATOR}, actuator $\Act{a}$ is a mapping from blocks of a partition of reactive space to 
a transduction space. Suppose $(B, \nu_B) \in \Act{a}$. We desire to associate the transduction space point 
$\nu_B$ not with its containing block $B$, but rather with a point $\psi \in B$ in reactive space. Thus we 
re-introduce the containing block conversion of definition \ref{D:POINT_BLOCK_CONVERSION}, which maps points 
of reactive space to their containing blocks of the partition. The actuator maps 
$\Act{a} \colon \#(\Closure{\Psi}) \to N$, and the containing block conversion maps 
$\blocks_{\domain{\Act{a}}} \colon \Closure{\Psi} \to \#(\Closure{\Psi})$. Thus the composite function maps 
$(\Act{a} \circ \blocks_{\domain{\Act{a}}}) \colon \Closure{\Psi} \to N$. For any block $B \in \#(\Closure{\Psi})$ and for 
any reactive stimulus $\psi \in B$, $(\psi, \nu_B) \in \Act{a} \circ \blocks_{\domain{\Act{a}}}$  as desired.
\end{proof}

\begin{definition}[transducer] \label{D:TRANSDUCER}
Let $\Act{a} \colon \#(\Closure{\Psi}) \to N$ be an actuator per definition \ref{D:ACTUATOR}, 
and let $\blocks_{\domain{\Act{a}}} \colon \Closure{\Psi} \to \#(\Closure{\Psi})$ be its containing block function 
(definition \ref{D:POINT_BLOCK_CONVERSION}).
A \emph{transducer} is a re-indexed actuator \\ 
$\Act{a} \circ \blocks_{\domain{\Act{a}}} \colon \Closure{\Psi} \to N$ 
of lemma \ref{L:ACTUATOR_REINDEX}.
\end{definition}

\begin{remark} \label{D:TRANSDUCTION}
A \emph{transduction} applies a re-indexed actuator 
$\Act{a} \circ \blocks_{\domain{\Act{a}}} \colon \Closure{\Psi} \to N$ 
to infer the value $\nu \in N$ from a reactive stimulus.
\end{remark}

\section{Actuator network and labels} \label{S:ACTUATOR_IDENTIFIER} 
Labels facilitate identifying the current position in the automaton's stepwise network of decisions (actuators) and 
consequent actions (functionalities). To this end, each actuator is assigned a unique identifier called a 
\emph{locus}. A locus is simply a location identifier (address) for an actuator.

\begin{definition} \label{D:LOCUS}
A \emph{locus} is a label for an actuator.
\end{definition}

The impetus for this term is its use as a target of a \Dquo{goto} statement in programming languages. The interrogative 
\Dquo{goto where?} implies that a location (locus) is the desired response. Loci become the foundation for connectivity 
in a network of actuators.

\section{Space of transduction} 
Consider an actuator, which is a mapping between the blocks of a partition of reactive space 
(\Dquo{positions}) and an otherwise unrelated set of values. Suppose the mapping is 
multi-valued, consisting of 
\begin{itemize}
  \item the current functionality, which transforms the current reactive stimulus into the   
        current response, and
  \item the next locus, which is the basis for transition in a network.
\end{itemize}

In any reactive actuated automaton, the transduced value of each actuator consists of an ordered pair: 
the current functionality and the next locus.  

\begin{definition} \label{D:TRANSDUCTION_SPACE}
The \emph{transduction} space of an actuator is the set $\Fty{F} \times \Lambda$. 
A member $(\Ftylc{f}, \lambda')$ is a \emph{transduced} value.
\end{definition}

\begin{theorem}[the \emph{transduced} values] \label{T:TRANSDUCED_VALUE}
The values transduced from actuator $\Act{a}$, applied at reactive stimulus $\psi$, are:  
\begin{align*}
    (\Ftylc{f}, \lambda') &= \transduction{\Act{a}}{\psi} \in \Fty{F} \times \Lambda \\ 
    \text{or, individually} &: \\ 
    \Ftylc{f} &= \tranfunc{\Act{a}}{\psi} \in \Fty{F} \\
    \Ftylc{f}(\psi) &= (\tranfunc{\Act{a}}{\psi})(\psi) \in \textstyle{\Closure{\Phi}} \\
    \lambda' &= \tranaddr{\Act{a}}{\psi} \in \Lambda.
\end{align*}
\end{theorem}
\begin{proof}
Actuation from reactive stimulus means the re-indexing of an actuator with respect to reactive 
stimulus (lemma \ref{L:ACTUATOR_REINDEX}). Let $\blocks_{\domain{\Act{a}}}$ be the block 
conversion mapping for $\domain{\Act{a}}$, the partition of actuator $\Act{a}$ (see 
\S\ref{S:PARTITIONS} and definition \ref{D:POINT_BLOCK_CONVERSION}). Applying this lemma to 
the actuator $\Act{a}$, the re-indexed actuator is 
$(\Ftylc{f}_{i}, \lambda') = \transduction{\Act{a}}{\psi}$. 

The remaining formulas are straightforward manipulations of the abscissa, ordinate, and application operations.
\end{proof}

Consider a transducer $\tranbasic{\Act{a}} \colon \Closure{\Psi} \to \Fty{F} \times \Lambda$
(re-indexed actuator, definition \ref{D:TRANSDUCER}). 
This means that $(\Ftylc{f}, \lambda') \in \Fty{F} \times \Lambda$ is the transduced value 
(theorem \ref{T:TRANSDUCED_VALUE}).
Functionality $\Ftylc{f}$ completes frame $\Frm{f} = (\psi, \Ftylc{f}(\psi))$, and 
$\lambda'$ becomes the locus of the next actuator.

\begin{remark}
The number of transductions is finite. One must not mistake the cardinality of transductions for the cardinality of 
arguments to the transducer $\transduction{\Act{a}}{\phi}$, which is $\Card{\Closure{\Psi}}$. The variety of transductions 
is limited by Cauchy-Schwarz, $\Card{\Fty{F} \times \Lambda} \leq \Card{\Fty{F}}\Card{\Lambda}$. 
\end{remark}

\chapter{Catalogs} \label{S:CATALOG_CHAPTER} 

Discrete systems theory (software) is identified with the reactive actuated automaton (\Acro{RAA}). 
The term \Dquo{discrete} designates that each automaton step is associated with a finite dead time, during 
which no event is recorded.

An automaton is a self-governing machine whose architecture of stimuli and responses automates 
an algorithm. The deterministic finite automaton (\Acro{DFA}, see example in appendix \ref{S:DFA}) is a 
simple structure describing transit-based behavior. However, the \Acro{DFA} does not explain 
the underlying physics by which transitions are accomplished. The \Acro{DFA} can be modified to make explicit the 
mechanism governing state transition. The result is the reactive actuated automaton, which mechanizes logic 
using structure analogous to programming language. An informal analogy between a reactive actuated automaton 
and a programming language will be proposed at the end of this section.

\section{Inventory of catalogs} \label{S:CATALOG_INVENTORY} 
The reactive actuated automaton is constructed from intertwined sets of sets. 
Each of these sets is given the special name \emph{catalog}. 

\subsection{Catalog of reactivity} \label{S:STIMULUS} 
The catalog of reactivity is similar to the previous definition of \emph{reactive basis} 
(definition \ref{D:REACTIVE_BASIS}). Accurately $(\Phi, \Xi)$ is a catalog while its 
closure $(\widehat{\Phi}, \widehat{\Xi})$ is a reactive basis.

\subsection{Catalog of functionalities} \label{S:CATALOG_OF_FUNCTIONALITY} 
Functionalities were introduced earlier in \S\ref{S:FUNCTIONALITY}.

\begin{definition}[catalog of functionality] \label{D:CATALOG_OF_FUNCTIONALITY}
Let $\Basis{\Phi}{\Xi}$ be a reactive basis with total stimulus 
$\Closure{\Psi} = \Closure{\Phi}\Closure{\Xi}$. A finite subset 
$\Fty{F} \subseteq {\Closure{\Phi}}^{\Closure{\Psi}}$ is a \emph{catalog of functionality}.
\end{definition}

\begin{definition}[procedure] \label{D:PROCEDURE}
Let $\Basis{\Phi}{\Xi}$ be a reactive basis and suppose 
$\Fty{F} \subseteq {\Closure{\Phi}}^{\Closure{\Psi}}$ is a catalog of functionality. 
A \emph{procedure} is a sequence 
$\lbrace \Ftylc{f}_n \rbrace \colon \naturalnumbers \to \Fty{F}$.
\end{definition}

\subsection{Catalogs of loci and actuation} \label{S:CATALOG_OF_ACTUATION} 
Since one is a label for the other, the catalogs of loci and actuation are interrelated.

\begin{definition}[catalog of actuation] \label{D:CATALOG_OF_ACTUATION}
The catalog of \emph{actuation} $\Act{A} \subseteq N^{\#(\Closure{\Psi})}$ is a finite subset of the collection of actuators.
\end{definition}

\begin{definition}[catalog of loci] \label{D:CATALOG_OF_LOCI}
A catalog of \emph{loci} is a set $\Lambda$, each element $\lambda$ of which 
bijectively identifies some member of the catalog of actuation.
\end{definition}

\begin{notation}[anonymous label bijection] \label{N:ANONYMOUS_BIJECTION}
By hypothesis, there exists an anonymous bijection $\Lambda \to \Act{A}$ between loci and actuators. Temporarily 
allow the symbol \Dquo{$\lozenge$} to stand for this bijection. The application $\lozenge(\lambda)$ is 
represented in de-referencing notation as $\Bigast{\lambda}$. 
\end{notation}

\begin{remark}
The symbol for the anonymous label bijection will be used only in this section, 
\S\ref{S:CATALOG_OF_ACTUATION}. 
The notation $\Bigast{\lambda}$ will be preferred  because it is already familiar as the de-referencing symbol 
in the C programming language.
\end{remark}

\begin{theorem} \label{T:CATALOG_OF_ACTUATION_APPLICATION}
The application of the anonymous label bijection ($\lozenge$) at $\lambda \in \Lambda$ 
is an actuator $\Bigast{\lambda}$.
\end{theorem}
\begin{proof}
The anonymous label bijection $\lozenge \colon \Lambda \to \Act{A}$ is a mapping, and 
$\Act{A} \subseteq N^{\#(\Closure{\Psi})}$, the set of actuators; hence $\lozenge(\lambda) \in N^{\#(\Closure{\Psi})}$. 
By notation \ref{N:ANONYMOUS_BIJECTION}, $\lozenge(\lambda) = \Bigast{\lambda}$.
\end{proof}

\begin{definition}
Let $\Lambda$ be a catalog of loci. A \emph{path} is a sequence of actuator labels
$\lbrace \lambda_n \rbrace \colon \naturalnumbers \to \Lambda$ \\(see McCabe, appendix \ref{S:LOCUS_DIGRAPH}).
\end{definition}

\chapter{Configuration} \label{S:CONFIGURATION_CHAPTER} 
The subject of the previous chapter was information structures called catalogs. 
The present topic is how these catalogs combine into a system.  

\section{Configuration space and configuration} \label{S:CONFIGURATION_SPACE} 
\begin{definition}[configuration space] \label{D:CONFIG_SPACE}
The \emph{configuration space} for a reactive actuated automaton 
$\Auto{A}$ consists of six synchronized catalogs: a reactive basis $\Basis{\Phi}{\Xi}$, 
along with catalog of actuation $\Act{A}$, catalog of functionality $\Fty{F}$, 
and catalog of loci $\Lambda$. This configuration space is the Cartesian product 
\[
\ConfigSpace.
\] 
\end{definition}


Automata exist in many varieties. Since the reactive actuated automaton occupies 
the entire present scope of interest, we forgo mandatory use of the qualifiers 
\Dquo{reactive} and \Dquo{actuated.\!}

\begin{definition} [configuration] \label{D:CONFIGURATION}
A configuration $\ConfigScript{c} = \ConfigStd$ is a member of the configuration space $\ConfigSpace$. 
\end{definition}

\section{Configuration and unique state} \label{S:CONFIGURATION_UNIQUE_STATE} 
\begin{axiom} \label{A:CONFIGURATION_UNIQUE_STATE} 
With each configuration is associated exactly one state.
\end{axiom}

\section{Projections and functions of a configuration} \label{S:PROJECTIONS} 
A configuration has several components; it's described by a whole-to-constituent relation.
The mappings that invert this relationship are called projections.

\begin{definition}[Projections] \label{D:CONFIGURATION_PROJECTION}
Let $\ConfigScript{c} = \ConfigStd$ be a member 
of configuration space $\ConfigScript{C} = \ConfigSpace$. 
A \emph{projection} is a mapping $\mho$ from the set of configurations $\ConfigScript{C}$ 
to one of its components. 
\end{definition}

Projections of particular interest are: \vspace{-0.3cm}
\begin{description}
  \item [locus projection] 
      $\mathbf{\mho_\Lambda \colon \ConfigScript{C} \to \Lambda}$: 
      set $\mho_\Lambda\ConfigStd = \lambda$.
  \item [functionality projection 
      $\mho_{\Fty{F}} \colon \ConfigScript{C} \to \Fty{F}$:] 
      set $\mho_\Fty{F}\ConfigStd = \Ftylc{f}$.
  \item [actuator projection 
      $\mho_{\Act{A}} \colon \ConfigScript{C} \to \Act{A}$:] 
      set $\mho_\Act{A}\ConfigStd = \Act{a}$.
\end{description}
\vspace{0.3cm}

\begin{definition}[frame function] \label{D:FRAME_MAPPING}
The \emph{frame} function (mapping) is 
$\mho_\Frm{F}\ConfigStd = (\phi\xi, \Ftylc{f}(\phi\xi)) = \Frm{f}$.
\end{definition}

\section{Connectedness of configurations} \label{S:CONNECTEDNESS} 
With each current configuration is associated a family of possible next configurations, 
dependent on the choice of the next stochastic stimulus. This next configuration has a 
complication: its value is uncertain before the current frame, and becomes certain only 
after its observation in the next frame (axiom \ref{A:EXIST_CONFIG}). 

\begin{definition}[parametric family of current configurations] \label{D:CURRENT_CONFIGURATION}
The set of current configurations is a three-parameter family $\ConfigScript{K}_\Curr$ of configurations:
\[
\ConfigScript{K}_\Curr(\phi, \xi, \lambda) 
    = (\phi, \xi, \Ftylc{f}_\Curr, \Act{a}_\Curr, \lambda) 
    = (\phi_\Curr, \xi_\Curr, \Ftylc{f}_\Curr, \Act{a}_\Curr, \lambda_\Curr)
\]
where:
\begin{align*}
    (\phi_\Curr, \xi_\Curr, \lambda_\Curr) &= (\phi, \xi, \lambda) 
                                               \in \ConfigSpace
                                           && \textrm{(free parameters)} \\
                          \Act{a}_\Curr &= \Bigast{\lambda} 
                                           && \textrm{(dereference of current actuator)} \\
                        \Ftylc{f}_\Curr &= \tranfunc{\Act{a}_\Curr}{\phi\xi} 
                                           && \textrm{(functionality transduced from current actuator)} \\
\end{align*}
\end{definition}
\emph{Rationale.}
Lemma \ref{L:ACTUATOR_REINDEX} (actuator re-indexing), 
notation \ref{N:ANONYMOUS_BIJECTION} (anonymous label bijection), 
theorem \ref{T:TRANSDUCED_VALUE} (the transduced values),  
and embedded side comments.

\begin{definition}[parametric family of next configurations] \label{D:NEXT_CONFIGURATION}
The set of next configurations is a four-parameter family $\ConfigScript{K}_\Next$ of configurations:
\[
{\ConfigScript{K}_\Next}(\phi, \xi, \lambda, \xi') = 
    (\phi_\Next, \xi_\Next, \Act{a}_\Next, \Ftylc{f}_\Next, \lambda_\Next)
    \in \textstyle{\ConfigSpace}
\]
where:
\begin{align*}
                   \phi, \xi, \lambda, \xi' &&& \textrm{(free parameters)} \\
                           \Act{a}_\Curr &= \Bigast{\lambda} 
                                            && \textrm{(dereference of current actuator)} \\
                         \Ftylc{f}_\Curr &= \tranfunc{\Act{a}_\Curr}{\phi\xi} 
                                            && \textrm{(functionality transduced from current actuator)} \\
                           \lambda_\Next &= \tranaddr{\Act{a}_\Curr}{\phi\xi} 
                                            && \textrm{(next locus transduced from current actuator)} \\
                           \varphi_\Curr &= \Ftylc{f}_\Curr(\phi\xi) 
                                            && \textrm{(response from current actuator)} \\
                              \phi_\Next &= \varphi_\Curr  = \Ftylc{f}_\Curr(\phi\xi) 
                                            && \textrm{(feedback from current response to next stimulus)} \\
                               \xi_\Next &= \xi' \\
                           \Act{a}_\Next &= \Bigast{\lambda_\Next} 
                                            && \textrm{(dereference of next actuator)} \\
                         \Ftylc{f}_\Next &= \tranfunc{\Act{a}_\Next}{\phi_\Next\xi_\Next} 
                                            && \textrm{(functionality transduced from next actuator)} \\
\end{align*}
\end{definition}
\begin{rationale}
Lemma \ref{L:ACTUATOR_REINDEX} (actuator re-indexing), notation \ref{N:ANONYMOUS_BIJECTION} (anonymous 
label bijection), theorem \ref{T:TRANSDUCED_VALUE} (the transduced values), and embedded side comments.
\end{rationale}

\begin{definition}[connected configurations] \label{D:CONNECTED_CONFIGURATIONS}
Let $\ConfigScript{c}_1 = \ConfigSub{1}$ and $\ConfigScript{c}_2 = \ConfigSub{2}$ be two 
configurations. If 
$\ConfigScript{c}_1 \in {\mathfrak{K}_\Curr}(\phi_1, \xi_1, \lambda_1)$ and 
$\ConfigScript{c}_2 \in {\mathfrak{K}_\Next}(\phi_1, \xi_1, \lambda_1, \xi_2)$, then 
$\ConfigScript{c}_1$ is directionally connected to $\ConfigScript{c}_2$. 
\end{definition}

\begin{notation}
The notation $\ConfigScript{c}_1 \prec \ConfigScript{c}_2$ signifies 
\Dquo{directionally connected} and
$\ConfigScript{c}_1 \nprec \ConfigScript{c}_2$ otherwise.
\end{notation}

\begin{remark}
Connectivity of configurations ($\prec$) is not generally transitive.
\end{remark}

\section{Trajectory} \label{S:TRAJECTOY} 

\begin{definition}[walk] \label{D:WALK}
A \emph{walk} is a sequence of configurations 
$\naturalnumbers \to \ConfigSpace$.
\end{definition}

\begin{definition}[trajectory] \label{D:TRAJECTORY}
A walk $\Idiom{\ConfigScript{c}}{n}$, each configuration $\ConfigScript{c}_i$ of which is 
directionally connected to 
the following configuration $\ConfigScript{c}_{i+1}$, is a \emph{trajectory}  
(compare definition \ref{D:SEQUENTIALLY_COUPLED}).
\end{definition}

\begin{nomenclature}
The term \Dquo{run} is synonymous with \emph{trajectory}. 
\end{nomenclature}

\begin{remark}
It is possible that a particular configuraton appears in one run but not in 
another.
\end{remark}

\begin{definition}[finite run] \label{D:FINITE_RUN}
An initial segment of a run (trajectory).
\end{definition}

\section{Inductive behavior of automaton runs} \label{S:INDUCTIVE_AUTOMATON} 
A reactive actuated automaton $\Auto{A}$ is a class of configurations 
in $\ConfigSpace$. A \emph{run} of automaton $\Auto{A}$ is also dependent on a 
\emph{starting} configuration $\ConfigScript{c}_0$. 

\begin{definition}[automaton run]
An automaton \emph{run} $\Idiom{\ConfigScript{c}}{n}$ is an inductive sequence built 
from connected members of the configuration space $\ConfigSpace$: 
\vspace{-8pt}
\begin{description}
  \item[Base Clause:] For some (initial) configuration $\ConfigScript{c}_0$, 
       $\ConfigScript{c}_0 \in \mathfrak{A}$. 
  \item[Inductive Clause:] For any configurations $\ConfigScript{c}$ 
       and $\ConfigScript{c}'$, 
       $\ConfigScript{c} \in \mathfrak{K}_\Curr(\phi, \xi, \lambda)$ 
       (theorem \ref{D:CURRENT_CONFIGURATION}) and \\
       $\ConfigScript{c}' \in \mathfrak{K}_\Next(\phi, \xi, \lambda, \dot{\xi})$ 
       (theorem \ref{D:NEXT_CONFIGURATION}), 
       $\ConfigScript{c}' \in \mathfrak{A}$ if and only if $\ConfigScript{c} \in \mathfrak{A}$.
\end{description}
\end{definition}

This inductive class of configurations is based on an intersection of the two parametric 
families: the current family of configurations, and the next family of configurations 
.

\chapter{State spaces} \label{S:STATE_SPACES} 
A state space is a proper subset of the configuration of the 
reactive actuated automaton that carries information 
equivalent to the entire configuration. 
By \Dquo{carries information} we shall mean that the full set is 
algebraically recoverable from the subset. 
We will identify two varieties of state space: indigenous and 
triune step. Triune step state space is inspired by explicit operation of 
the automaton's configuration space, while indigenous state space 
considers a minimal sufficient subset of the configuration space. 

Regardless of which state space is used, the reactive actuated 
automaton has the same configuration space. 

\section{Triune step state space} \label{S:STEP_STATE_SPACE} 
While leaving undefined the notion of a step in an algorithm, we 
do formalize it for the automaton, where a triune step is equivalent to three 
individual points: 
one from a \emph{path} (definition \ref{D:}).
one from a \emph{process} (definition \ref{D:}).
and one from a \emph{procedure} (definition \ref{D:}).

A triune step consists of processing an actuator, which induces a 
functionality and a transition 
to the next actuator. This information is included in the 
Cartesian product of the catalogs of loci, functionality, 
and frames. 

\begin{definition}[triune step space] \label{D:STEP_SPACE}
Let $\Lambda$ be a catalog of loci, and suppose basis 
$\Basis{\Phi}{\Xi}$ underlies the catalog of functionality 
$\Fty{F} \subseteq {\Closure{\Phi}}^{\Closure{\Phi}\Closure{\Xi}}$ 
and the frame space 
$\Frm{F} = \Closure{\Phi}\Closure{\Xi} \times \Closure{\Phi}$. A 
\emph{triune step space} $\Stp{S}$ is the Cartesian product 
$\Stp{S} = \Lambda \times \Fty{F} \times \Frm{F}$.
\end{definition}

\begin{nomenclature}[triune step]
A triune step state is a point in triune step space.
\end{nomenclature}

\begin{theorem} \label{T:REGENERATE_FROM_STEP}
The configuration $\ConfigStd \in \ConfigSpace$ is algebraically 
regenerable from the triune step state $\Stplc{s} = (\lambda, \Ftylc{f}, \Frm{f})$. 
\end{theorem}
\begin{proof}
The values of $\lambda \in \Lambda$, $\Ftylc{f} \in \Fty{F}$, and 
$\Frm{f} \in \Frm{F}$ are given by hypothesis. 
Other values may be read off directly. The value of actuator 
$\Act{a} \in \Act{A}$ is $\Act{a} = \Bigast{\lambda}$ 
(see notation \ref{N:ANONYMOUS_BIJECTION}). The total reactive 
stimulus $\psi$ is $\absc {\Frm{f}}$, which has 
dyadic constituents $\phi = \Restrict{\absc {\Frm{f}}}{\domain \Phi} \in \Closure{\Phi}$ and 
$\xi = \Restrict{\absc {\Frm{f}}}{\domain \Xi} \in \Closure{\Xi}$. 
\end{proof}

\section{Indigenous state space} \label{S:INDIGENOUS_STATE_SPACE} 
The reactive actuated automaton's configuration space $\ConfigSpace$ is interlocked.
Indigenous state space is one answer to which components comprise a minimum sufficient set.  

\begin{definition}[indigenous space] \label{D:INDIGENOUS_STATE_SPACE}
\emph{Indigenous} state space $H$ is the Cartesian product 
$\Lambda \times \Closure{\Psi}$ of the catalog of loci and
the space of reactive stimulus.
\end{definition}

Members of indigenous space are called indigenous states.
Indigenous states compactly summarize the operational condition of an automaton, and map to equivalent steps.

\begin{theorem} \label{T:RECOVER_FROM_INDIGENOUS}
The configuration $\ConfigStd \in \ConfigSpace$ is algebraically recoverable from the indigenous state $\eta = (\lambda, \psi)$. 
\end{theorem}
\begin{proof}
The values of $\lambda \in \Lambda$ and $\psi \in \Closure{\Psi}$ are given by hypothesis. The value of actuator $\Act{a} \in \Act{A}$ is $\Act{a} = \Bigast{\lambda}$ (see 
notation \ref{N:ANONYMOUS_BIJECTION}). 
The given reactive stimulus $\psi$ has dyadic constituents $\phi = \Restrict{\psi}{\domain \Phi} \in \Closure{\Xi}$ 
and $\xi  = \Restrict{\psi}{\domain \Xi} \in \Closure{\Xi}$. 
The functionality $\Ftylc{f}$ is transduced from actuator $\Act{a}$: 
$\Ftylc{f} = \tranfunc{\Act{a}}{\psi} \in \Fty{F}$ (see theorem \ref{T:TRANSDUCED_VALUE}).
\end{proof}

\chapter{The reactive actuated automaton} \label{S:RAA_CHAPTER} 

\section{Principle of operation} \label{S:RAA_OPERATION} 
Operation of the \Acro{RAA} is recognized as the base and repetitive clauses of an 
inductive form:
\begin{enumerate}
  \item Initialize. [definition \ref{D:INITIALIZER}]
  \item Cycle. [definition \ref{D:CYCLE}]
  \item Go To 2. [Repeat]
\end{enumerate}

\emph{Initialize} is the base clause of the inductive form, which prepares the 
automaton's context. \emph{Cycle} is the organization of \Acro{RAA} subtasks for the 
repetitive clause. 

\begin{remark}
Before undertaking initialization, it is necessary to examine the cycle to discover 
and classify its sensitivities.
\end{remark}

\subsection{Cycle} \label{S:RAA_CYCLE} 
\emph{Cycle} is the organization of \Acro{RAA} subtasks for the inductive clause. 

\begin{definition}[cycle] \label{D:CYCLE} 
A \emph{cycle} of a reactive actuated automaton is a unit of work consisting of several ordered subtasks:
\begin{enumerate}
  \item Observe the current values of the stochastic stimuli \\
        $\xi_\Curr = \Instance{\Closure{\Xi}}$ \quad [random sample].
  \item Update the current stimulus as the union of the current stochastic stimulus and 
        the previous response \\ 
        $\psi_\Curr = \varphi_\Prev \cup \xi_\Curr$ \quad [feedback, \S\ref{S:FEEDBACK}].
  \item Update the current locus as the previous cycle's next actuator locus \\ 
        $\lambda_\Curr = \lambda_\Prev$.
  \item Determine actuator addressed by current locus \\
        $\Act{a}_\Curr = \Bigast{\lambda}_\Curr$
        \quad [dereference notation \ref{N:ANONYMOUS_BIJECTION}]. 
  \item Determine the transduced quantities (current functionality and next locus) \\
        $(\Ftylc{f}_\Curr, \lambda_\Next) = \transduction{\Act{a}_\Curr}{\psi_\Curr}$
        \quad [theorem \ref{T:TRANSDUCED_VALUE}].
  \item Determine the current response by applying the current functionality to the current reactive stimulus \\
        $\varphi_\Curr = \Ftylc{f}_\Curr(\psi_\Curr)$. 
\end{enumerate}
\end{definition}

\subsection{Initialize} \label{S:RAA_INITIALIZATION} 
It is evident that \emph{cycle} goes awry whenever $\varphi_\Prev$ (step 2) or 
$\lambda_\Prev$ (step 3) is undefined. On the first cycle, these are not defined, 
so $\psi_\Curr$ and $\lambda_\Curr$ are not updated. \emph{Initialize}, 
the base clause of the inductive form, must establish reasonable defaults. 

\begin{definition}[initializer] \label{D:INITIALIZER} 
Let $\Idiom{\ConfigScript{c}}{n}$ be the trajectory to be initialized. 
\vspace{-0.25cm} 
\begin{itemize} 
  \item If there exists a configuration $\ConfigScript{c}$ such that
        $\ConfigScript{c} \prec \ConfigScript{c}_1$, 
        set actuator $\Act{a} = \mho_{\Act{A}}(\ConfigScript{c})$. 
  \item Compute the transduced values 
        $(\varphi_\Prev, \lambda_\Prev) = \transduction{\Act{a}}{\psi}$
        \quad\quad\quad
        (theorem \ref{T:TRANSDUCED_VALUE}). 
\end{itemize} 
\vspace{-0.25cm} 
The value $(\varphi_\Prev, \lambda_\Prev)$ is a valid initializer for $\ConfigScript{c}_1$.
\end{definition}

\section{Halting} \label{S:HALTING} 
\begin{theorem}[reactive systems do not halt] \label{T:NO_HALT}
Regardless of stimulus, in a system the next configuration is completely and uniquely defined 
through next actuator. 
\end{theorem}
\begin{proof}
Regardless of stimulus, the next configuration in a system is completely and uniquely defined 
through next actuator. 
\end{proof}

\section{Sketch of programming language analogy} \label{S:LANGUAGE_ANALOGY} 
The reactive actuated automaton is a network of actuators. 
Each actuator may possess many alternate functionalities and successor actuators. 

The automaton configuration space $\ConfigSpace$ of definition \ref{D:CONFIG_SPACE} resembles an elementary programming 
language. The following comprise an oversimplified analogy:
\begin{itemize}
  \item the domain of deterministic class $\Phi$ represents ordinary program variables;
  \item the domain of stochastic class $\Xi$ represents read-only external inputs \\
        (e.g. sensors, remote commands);
  \item actuators in class $\Act{A}$ represent if-then-elsif-else contingencies;
  \item the transduced functionality $\Ftylc{f}$ represents a reactive stimulus that modifies deterministic variables; 
        \\and 
  \item the transduced label $\lambda$ is a \Dquo{goto} indicating the next actuator to be exercised.
\end{itemize}

\chapter{Reactive iteration} \label{S:ITERATION_CHAPTER} 
Iteration is the transition of an automaton from its current configuration to the next 
configuration, which is unique for deterministic phenomena. Reactive iteration generalizes automorphic iteration; the difference lies in stochastic stimuli. 

In automorphic iteration, all information required to complete an iteration resides in
deterministic stimuli. This condition is not true of reactive iteration. 

\section{Reactivity} \label{S:REACTIVE_ITERATION} 
Reactivity concerns response to stimuli. 

\begin{nomenclature}[observation]
\emph{Observation} of a stochastic stimulus is the term used for determining its 
instantaneous value.
\end{nomenclature}

\begin{axiom}[value of stochastic stimulus]
The value of a stochastic stimulus is observable at any instant of time. 
\end{axiom}
\begin{remark}
This law does not preclude that the observed value of a stochastic variable can be 
copied into a deterministic variabe for later reference.
\end{remark}

For the reactive actuated automaton, the status of a collection of stochastic 
stimuli is simultaneously observable at any instant. 
This occurs as step 1 of the \emph{cycle} sub-process (definition \ref{D:CYCLE}).  

\begin{axiom}[stochastic-deterministic confusion]
The value of any stochastic stimulus cannot be set or influenced by any deterministic stimuli. 
\end{axiom}

\section{Iteration} \label{S:ITERATION} 
Iteration involves at least two adjacent cycles (definition \ref{D:CYCLE})
and creates a trajectory segment (definition \ref{D:TRAJECTORY}). 

\section{Summary of forward reactive iteration} \label{S:REVIEW_ITERATION} 

The following two-cycle table is summarized from the explicit iteration formulae of 
definitions \ref{D:CURRENT_CONFIGURATION}, \ref{D:NEXT_CONFIGURATION}, and theorem 
\ref{T:TRANSDUCED_VALUE} 
 
\begin{table}[h!]
    \begin{center}
        \begin{tabular}{|c|c|c|}
            \hline
Cycle&Configuration Item&Formula \\ \hline
current&$\phi_\Curr$&\textrm{(free parameter)} \\
&$\xi_\Curr$&$\Instance{\widehat{\Xi}}$ \textrm{pseudo-funtion}  \\
&$\lambda_\Curr$&\textrm{(free parameter)} \\
&$\Act{a}_\Curr$&$\Bigast{\lambda_\Curr}$ \\
&$\Ftylc{f}_\Curr$&$\tranfunc{\Act{a}_\Curr}{\phi_\Curr\xi_\Curr}$ \\ \hline
next&$\phi_\Next$&$\varphi_\Curr  = \Ftylc{f}_\Curr(\phi_\Curr\xi_\Curr)$ \quad[feedback] \\
&$\xi_\Next$&$\Instance{\widehat{\Xi}}$ \textrm{pseudo-function} \\
&$\Act{a}_\Next$&$\Bigast{\lambda_\Next}$ \\
&$\Ftylc{f}_\Next$&$\tranfunc{\Act{a}_\Next}{\phi_\Next\xi_\Next}$ \\ 
&$\lambda_\Next$&$\tranaddr{\Act{a}_\Next}{\phi_\Next\xi_\Next}$ \\ 
\hline
        \end{tabular}
    \end{center}
    \caption{Configuration at end of cycle} \label{Ta:CONFIG_FIRST_CYCLE}
\end{table}

\section{Reverse iteration} \label{S:REVERSE_ITERATION_SECTION} 
Our construction of automata has provided that configurations unfold sequentially~-- 
that is, the next configuration becomes known after completing the current configuration. 
Consequently automata inherit an intrinsic \Dquo{forward} orientation. It is also 
reasonable to inquire what configuration may have occurred previously. This question is the 
motivation for reverse iteration, which considers automata operating backwards.

\subsection{Reverse inference} \label{S:ALTERNATIVE_SECTION} 
Suppose \Dquo{inference} is the task of determining what configuration must follow another 
in a trajectory. Because the reactive actuated automaton is a system, this configuration 
must be unique. (see axiom \ref{A:CONFIGURATION_UNIQUE_STATE})

Consequently, reverse inference is identifying all immediate predecessor steps 
${\ConfigScript{c}_{n-1}}$ such that $\ConfigScript{c}_n = \Auto{A}(\ConfigScript{c}_{n-1})$. We have 
\[
{\Auto{A}}^{-1}(\ConfigScript{c}) = \lbrace \tilde{\ConfigScript{c}} \in \ConfigScript{C} \;\colon\; \Auto{A}(\tilde{\ConfigScript{c}}) = \ConfigScript{c} \rbrace.
\]

\begin{remark}
Although referring informally to the inverse ${\Auto{A}}^{-1}$ of automaton $\Auto{A}$, 
speaking precisely we have defined the inverse $\tilde{\ConfigScript{c}}$ of a 
\emph{configuration} $\ConfigScript{c} \in \ConfigScript{C}$ within the configuration space.
\end{remark}

\subsection{Reverse inference as subset of configuration space} \label{S:POWER_FUNCTION_INVERSE} 
\begin{definition} \label{D:STOCHASTIC_INVERSE}
Let $\ConfigScript{c} \in \ConfigScript{C}$ be a configuration and configuration space, 
Reverse inference $u^{-1}(\ConfigScript{c})$ is the 
mapping $\AutoInv{a} \colon \ConfigScript{C} \to \powerset{\ConfigScript{C}}$ defined by 
\[
u^{-1}(\ConfigScript{c}) = \lbrace \tilde{\ConfigScript{c}} \in \ConfigScript{C} 
\colon u(\tilde{\ConfigScript{c}}) = \ConfigScript{c} \rbrace,
\] 
where $\powerset{\ConfigScript{C}}$ denotes the power set of $\ConfigScript{C}$. 
\end{definition}

In other words, the reverse inference relation is a mapping from a configuration to a 
set of configurations.

\section{Reactive morphism} \label{S:REACTIVE_MORPHISM} 
\begin{definition}[reactive morphism] \label{D:REACTIVE_MORPHISM} 
Let $\ConfigScript{C}$ be the set of all configurations and $\widehat{\Xi}$ be the included 
set of stochastic stimuli. A \emph{reactive} morphism is a mapping 
$v: \ConfigScript{C} \to \ConfigScript{C}^{\widehat{\Xi}} 
$, where 
$v\ConfigSub{1} = (\xi_2,\ConfigSub{2})$, and also $\xi_2$ (first instance) equals 
$\xi_2$ (second instance). 
\end{definition}

\begin{remark}
$(\xi_2,\ConfigSub{2})$ is a term of the mapping $\ConfigScript{C}^{\widehat{\Xi}}$ 
in which $\xi_2$ appaers twice. 
\end{remark}

The work to determine $\ConfigSub{1} \prec \ConfigSub{2}$ 
is implicit in the reactive morphism $v$ (definition \ref{D:CONNECTED_CONFIGURATIONS}). 

\subsection{Reactive morphism equality constraint} \label{S:STOCHASTIC_EQUALITY} 
Let $\ConfigScript{C}$ be the set of configurations and $\widehat{\Xi}$ be the  
set of stochastic stimuli. The sets overlap in the sense $\widehat{\Xi} \in \ConfigScript{C}$. 

This can possibly introduce ambiguity, as illustrated 
by the expansion to elementary terms of $(\xi_2,\ConfigSub{2})$ of the expression 
$\ConfigScript{C}^{\widehat{\Xi}}$. The ambiguity is whether the duplicates of 
$\xi_2$ are free or bound. 

\begin{definition}[reactive morphism equality constraint] \label{D:}
Definition \ref{D:REACTIVE_MORPHISM} states that a reactive morphism is a mapping 
$v: \ConfigScript{C} \to \ConfigScript{C}^{\widehat{\Xi}}$ (or 
$v: \ConfigScript{C} \to (\widehat{\Xi} \to \ConfigScript{C})$). It is understood that 
the repeated value of $\xi \in \widehat{\Xi}$ must equal the value embedded within 
$\ConfigScript{C}$. This condition is here called the \emph{reactive morphism equality} onstraint.
\end{definition}

\subsection{Inverse of reactive morphism} \label{S:ITERATIVE_INVERSE} 
The space of stochastic stimuli is part of configureation. Over the configuration 
space, the reactive morphism takes the abstract form 
$\ConfigScript{C} \to \ConfigScript{C}^{\widehat{\Xi}}$. A reactive morphism has 
abscissa $\ConfigSub{1}$ and ordinate $(\xi_2,\ConfigSub{2})$. 
According to set theory, the inverse of the reactive morphism has the 
\emph{ordinate-implies-abscissa} representation 
\[
(\xi_2,\ConfigSub{2}) \to \ConfigSub{1}.
\]

\begin{remark}
Both the reactive morphism and its inverse are mappings,
\end{remark}

%
%
%
%
%

\chapter{Cones} \label{S:CONE_CHAPTER} 
A cone is a construct prepared with the inverse of a reactive actuated automaton. 
It consists of all finite backwards trajectories converging to a given point.  

Although the term \Dquo{cone} is more ideologic than geometric, to preserve intuition this chapter uses triune step space rather than configuration. The two are equivalent 
(theorem \ref{T:REGENERATE_FROM_STEP}).

\section{Description} \label{S:CONE_DESCRIPTION} 
The reactive actuated automaton possesses a non-deterministic\footnote{That is, the \emph{inverse} is not generally a pointwise invertible 
mapping as often suggested by the term \emph{inverse}.} inverse relation. See \S\ref{S:ITERATIVE_INVERSE}.

A collection of reverse trajectories is realized through repetitive re-application of the inverse, converging to a designated \emph{crux} 
triune step. These iterative chains may be bounded (trimmed to finite length) by enforcing some stopping criterion. 
This construction results in the cone, a structured set of possible bounded trajectories eventually leading to the crux triune step.
The starting points of such trajectories are known as \emph{precursor} steps of the crux triune step.

\subsection{Predecessor generations} \label{S:PREDECESSOR_GENERATION} 
\begin{definition}
Hypothesize $\Auto{A}$ as a reactive actuated automaton containing configuration 
$\ConfigScript{c}_{0}$. 
Let $u \colon \ConfigScript{C} \to \ConfigScript{C}$ be a mapping with reverse inference relation 
$u^{-1} \colon \ConfigScript{C} \to \powerset{\ConfigScript{C}}$.

(base clause) 
Let base protoset $\ConfigScript{C}_\heartsuit^{(0)} = \lbrace \ConfigScript{c}_\text{crux} \rbrace$ be the 
$0^{\text{th}}$ predecessor generation of 
$\ConfigScript{a}_{0}$.

(inductive clause)
The ${(n+1)}^{\text{st}}$ generation predecessors are defined in terms of the 
$n^{\text{th}}$ generation:
\[
\ConfigScript{C}_\heartsuit^{(n+1)} = 
    \bigcup_{\ConfigScript{c} \in \ConfigScript{C}_\heartsuit^{(n)}} u^{-1}(\ConfigScript{c}).
\]
\end{definition}
\begin{remark}
This definition places protoset $\ConfigScript{C}_\heartsuit^{(1)} = u^{-1}(\ConfigScript{c}_{0})$.
\end{remark}
For a discussion of protoset $\ConfigScript{C}_\heartsuit^{(n)}$ in context of the Cartesian product, see Appendix \S\ref{D:PROTOSET}.

\section{Inductive definition of predecessor generations} \label{S:PREDECESSOR_GENERATION} 
\begin{definition}
Let $\Stp{S}$ be a triune step space containing crux step $\ConfigScript{a}_{\text{crux}}$. 
Suppose $u \colon \Stp{S} \to \Stp{S}$ is an automorphism with reverse inference relation 
$u^{-1} \colon \Stp{S} \to \powerset{\Stp{S}}$.

(base clause) 
Let base protoset $G_\heartsuit^{(0)} = \lbrace \ConfigScript{a}_\text{crux} \rbrace$ be the $0^{\text{th}}$ predecessor generation of 
$\ConfigScript{a}_{\text{crux}}$.

(inductive clause)
The ${(n+1)}^{\text{st}}$ generation predecessors are defined in terms of the 
$n^{\text{th}}$ generation:
\[
G_\heartsuit^{(n+1)} = 
    \bigcup_{\ConfigScript{a} \in G_\heartsuit^{(n)}} u^{-1}(\ConfigScript{a}).
\]
\end{definition}
\begin{remark}
This definition places protoset $G_\heartsuit^{(1)} = u^{-1}(\ConfigScript{a}_{\text{crux}})$.
\end{remark}
For a discussion of protoset $G_\heartsuit^{(n)}$ in context of the Cartesian product, see Appendix \S\ref{D:PROTOSET}.

\begin{definition} \label{D:ISOLATION} 
If $G_\heartsuit^{(n)}(\ConfigScript{a}_\text{crux}\}) = \varnothing$ for some natural number $n$, 
then $\ConfigScript{a}_\text{crux}$ is said to be \emph{isolated}.
\end{definition}
\begin{remark}
If $\ConfigScript{a}_\text{crux}$ is isolated, it means that its set of $n^\text{th}$ generation predecessors is empty. This set is 
without predecessors in terms of the inverse automaton, and is consequently unapproachable by forward steps of the parent automaton.
\end{remark}


\section{Predecessor trajectory} \label{S:CONE_WALK} 
A predecessor trajectory begins at triune step 
$\ConfigScript{a}_0 = \ConfigScript{a}_{\text{crux}}$ and proceeds backwards, indexing through 
the negative integers. We abuse the proper sense of the term \Dquo{sequence} by 
permitting an indexing not being the natural numbers.

\begin{definition}
Let $\Stp{S}$ be a triune step space. A bounded predecessor trajectory starting with step $\ConfigScript{a}_0 = \ConfigScript{a}_{\text{crux}}$ is a \emph{finite} 
sequence in triune step space such that $\ConfigScript{a}_{i-1} \prec \ConfigScript{a}_i$ for every $i \leq 0$.
\end{definition}
\begin{remark}
For example in the case $i = -2$, we have  $\ConfigScript{a}_{-3} \prec \ConfigScript{a}_{-2}$.
\end{remark}

Since a bounded predecessor trajectory $\Stplc{w}$ is a \emph{finite} sequence of steps, then it has a finite number of terms which run in
index $i$ from $-(n-1) \leq i \leq 0$, where $n = \Card{\Stplc{w}} < \infty$ is the number of triune steps in $\Stplc{w}$.

\begin{definition}  \label{S:COMPLETE_CONE_WALK}
A set $\Stp{W}$ of bounded predecessor trajectories, all starting at 
$\Stplc{w}_0 = \Stplc{s}_{\text{crux}}$, is \emph{complete} if\\
\[
\forall(\Stplc{w} \in \Stp{W}) 
\forall(-(\Card{\Stplc{w}}-2) \leq i \leq 0) 
\forall(\Stplc{s} \in v^{-1}(\Stplc{w}_{i}))
\exists(\Stplc{e} \in \Stp{W})
     : \Stplc{w}_{i} = \Stplc{e}_{i} \wedge \Stplc{s} = \Stplc{e}_{i-1}.
\]
\end{definition}

\begin{remark}
Completeness assures combinatorial diversity by preventing sub-walk duplication.
\end{remark}

\begin{definition} \label{S:DEPENDENT_CONE_WALK}
Let $\Stplc{w}$ and $\Stplc{w}'$ be localized predecessor walks starting at 
$\Stplc{w}_0 = \Stplc{w}_0' = \Stplc{s}_{\text{crux}}$. Suppose the length of $\Stplc{w}$ 
is $\Card{\Stplc{w}} = n$ and the length of $\Stplc{w}'$ is $\Card{\Stplc{w}'} = m$, with $m \leq n$. If $\Stplc{w}'(i) = \Stplc{w}(i)$ for every $-(m-1) \leq i \leq 0$, then $\Stplc{w}$ 
and $\Stplc{w}'$ are \emph{dependent} with $\Stplc{w}'$ \emph{dispensable}.
\end{definition}

\begin{definition} \label{S:INDEPENDENT_CONE_WALK}
Let $\Stp{W}$ be a set of localized predecessor walks starting at 
$\Stplc{s}_0 = \Stplc{s}_{\text{crux}}$. The set is \emph{independent} if it contains no 
dispensable member.
\end{definition}

\subsection{Cone} \label{S:CONE_CONE_WALK} 
\begin{definition} \label{S:CONE_DEFINITION}
A \emph{cone} $\Con{C}$ is a complete independent set of localized predecessor walks starting at $\Stplc{s}_{\text{crux}}$.
\end{definition}

\begin{remark}[Stopping rule]
We avoid the specificity of various stopping criteria (\S\ref{S:CONE_DESCRIPTION}) by introducing the equivalent but arbitrary 
notion of localization.
\end{remark}


\begin{definition} \label{S:INDEPENDENT_CONE_WALK}
Let $\Stp{W}$ be a set of bounded predecessor trajectories starting at $\Stplc{a}_0 = \Stplc{a}_{\text{crux}}$.
The set is \emph{independent} if it contains no dispensable member.
\end{definition}

\section{The cone formalism} \label{S:CONE_CONE_WALK} 
\begin{definition} \label{S:CONE_DEFINITION}
A \emph{cone} $\Con{C}$ is a complete independent set of bounded predecessor trajectories starting at non-isolated 
(definition \ref{D:ISOLATION}) triune step $\Stplc{a}_{\text{crux}}$.
\end{definition}

\begin{remark}
In definition \ref{S:CONE_DEFINITION}, it is profoundly difficult to detect the existential condition of non-isolation.
\end{remark}

\begin{remark}{(Stopping rule)}
We avoid the specificity of various stopping criteria (\S\ref{S:CONE_DESCRIPTION}) by introducing the equivalent but arbitrary 
notion of boundedness.
\end{remark}

\begin{definition} \label{S:CONE_EDGE_STEP}
Let $\Con{C}$ be a cone with $\Stplc{w} \in \Con{C}$ a member bounded predecessor trajectory. Suppose $n = \Card{\Stplc{w}}$ is the number of
steps in $\Stplc{w}$. The terminus $\Stplc{w}(-(n-1)) = \Stplc{w}_{-(n-1)}$ is the \emph{edge triune step} of trajectory $\Stplc{w}$.
\end{definition}

\begin{definition} \label{S:CONE_EDGE}
Let $\Con{C}$ be a cone.  Its \emph{edge}, written $\edge{\Con{C}}$, is the collection of edge steps of all member bounded 
predecessor trajectories.
\end{definition}

\begin{definition} \label{S:ACYLIC_CONE}
An \emph{acyclic} cone has no cycle (loop) in its path projection $\overline{\mho}_\Lambda$ 
(see \S\ref{D:SEQUENCE_PROJECTION}).
\end{definition}

\begin{theorem} \label{T:ACYLIC_CONE_CORRESPONDENCE}
The acyclic cone $\Con{C}$ and $\edge{\Con{C}}$ are in one-to-one correspondence via the edge triune step relation of a bounded 
predecessor trajectory.
\end{theorem}
\begin{proof}
Assume the opposite: there are different bounded predecessor trajectories with the same edge triune step. Let $\ConfigScript{u}$ and $\ConfigScript{v}$ be two 
different trajectories with common edge triune step $\Stplc{a}_\text{common}$.  

Suppose $\Card{\ConfigScript{u}} = m$ and $\Card{\ConfigScript{v}} = n$, so the 
indexes of $\Stplc{a}_\text{common}$ are $-(m-1)$ and $-(n-1)$ respectively. 

We assert that if $\ConfigScript{u}_{-(m-1)+i} = \ConfigScript{v}_{-(n-1)+i}$ for some $0 \leq i$, 
then $\ConfigScript{u}_{-(m-1)+(i+1)} = \ConfigScript{v}_{-(n-1)+(i+1)}$.

Suppose sequencing is governed by a reactive actuated automaton $\Auto{A}$. 
So sequenced, the next triune step in predecessor trajectory $\ConfigScript{u}$ is $\ConfigScript{u}_{-(m-1)+(i+1)} = \Auto{A}(\ConfigScript{u}_{-(m-1)+i})$.
Similarly, the next triune step in $\ConfigScript{v}$ is $\ConfigScript{v}_{-(n-1)+(i+1)} = \Auto{A}(\ConfigScript{v}_{-(n-1)+i})$.
But if $\ConfigScript{u}_{-(m-1)+i} = \ConfigScript{v}_{-(n-1)+i}$, then 
$\Auto{A}(\ConfigScript{u}_{-(m-1)+i}) = \Auto{A}(\ConfigScript{v}_{-(n-1)+i}) = \Auto{A}(\Stplc{a})$.
By transitivity of equality, $\ConfigScript{u}_{-(m-1)+(i+1)} = \Auto{A}(\Stplc{a}) = \ConfigScript{v}_{-(n-1)+(i+1)}$.

Without loss of generality suppose $m \leq n$. Then $\ConfigScript{u}_{-(m-1)+i} = \ConfigScript{v}_{-(n-1)+i}$ is true for $i = 0, 1, \;...\; m-1$.
At $i = m - 1$ we have $\Stplc{a}_{\text{crux}} = \ConfigScript{u}_0 = \ConfigScript{v}_{-(n-1)+(m-1)} = \ConfigScript{v}_{-(n-m)}$.
Since $\ConfigScript{v}$ is a bounded predecessor trajectory of a cone, then $\ConfigScript{v}_0 = \Stplc{a}_{\text{crux}}$.
But $\ConfigScript{v}_0 = \Stplc{a}_{\text{crux}} = \ConfigScript{v}_{-(n-m)}$. Because the cone is assumed acyclic, 
$\ConfigScript{v}_0$ and $\ConfigScript{v}_{-(n-m)}$ must then be the same identical triune step -- that is, $m = n$.

Here the assumption of two different bounded predecessor trajectories with the same edge triune step leads to the contradiction that both are indeed the 
same identical trajectory. This means bounded predecessor trajectories within an acyclic cone $\Con{C}$ are in one-to-one correspondence with 
$\edge{\Con{C}}$ via the edge triune step relation.
\end{proof}

\section{Hazards and multiple cone collections} \label{S:MULTIPLE_CONES} 
When a cone having a single crux is insufficient to encompass a given hazard, then a collection of cones is likely to suffice.
One example is the path-convergent family of cones:

\begin{definition} \label{D:MULTIPLE_CONES_PATH_CONVERGENT}
For some $\lambda_0 \in \Lambda$, each member $\Stplc{w}$ of a \emph{path-convergent} family of cones satisfies 
$\mho_\Lambda(\Stplc{w}_{\thinspace\text{crux}}) = \lambda_0$.
\end{definition}

\begin{remark}
A path-convergent collection does not necessarily include \emph{all} walks $\Stplc{w}$ such that 
$\mho_\Lambda(\Stplc{w}_{\thinspace\text{crux}}) = \lambda_0$.
\end{remark}

\chapter{Counting in trajectories} \label{S:COUNTING_CHAPTER} 
\begin{remark}
The following chapters frequently use the compound idiom that $\Idiom{x}{n}$ represents an 
anonymous sequence of objects of the same type as $x$. That is, if $X$ is the set of all 
$x_i$, then $\Idiom{x}{n} \colon \naturalnumbers \to X$.
\end{remark}

\section{Marked configurations} \label{S:COUNTING} 
Recall that definition \ref{D:CONFIGURATION} states that a configuration $\ConfigScript{c} = \ConfigStd$ is a member of the configuration space $\ConfigScript{C} = \ConfigSpace$.

Let $\Idiom{\ConfigScript{c}}{n}$ be a trajectory (infinite sequence of connected configurations) 
and let $Z$ be an arbitrary marked set of configurations. Simply summarized, 
$N_Z(\Idiom{\ConfigScript{c}}{n},k)$ counts the number of occurrences of any member of $Z$ 
before or at the $\xth{k}$ automaton configuration.

For those interested, details follow: 

\begin{definition}[marked set] \label{D:MARKED_SET}
A \emph{marked} set $Z$ is a finite subset of the collection of all configurations occurring in a trajectory $\Idiom{\ConfigScript{c}}{n}$.
\end{definition}

\begin{definition}[arrival] \label{D:ARRIVAL}
When the $\xth{i}$ configuration of the trajectory $\Idiom{\ConfigScript{c}}{n}$ is a member of 
$Z$ (that is, ${\ConfigScript{c}_i} \in Z$), then $\Idiom{\ConfigScript{c}}{n}$ is said to 
\emph{arrive} at $i$.
\end{definition}

\begin{definition}[arrival function] \label{D:ARRIVAL_FUNCTION}
An \emph{arrival} function is a sequence $\varphi = \{(1,n_1), (2,n_2), \cdots \}$ mapping 
each arrival, as identified by its ordinal occurrence number $i$, into its frame sequence 
number $n_i$. 
\end{definition}

This means that the first arrival occurs at frame sequence $n_1$, the second at $n_2$, 
\emph{et cetera}. The arrival function assumes the natural order, that is, $i < j$ implies 
$n_i < n_j$.

A related function counts \emph{how many} arrivals occur within a given duration: 

\begin{definition}[counting function] \label{D:COUNTING}
Suppose $\Idiom{\ConfigScript{c}}{n}$ is a trajectory and $Z$ is a set of arbitrarily marked 
configurations.
Let  $\varphi$ be an arrival function. The \emph{counting} function 
induced by the arrival function $\varphi$ is
\[
N_Z(\Idiom{\ConfigScript{c}}{n}, k) = \max \lbrace i \colon \varphi(i) \leq k \rbrace,
\] 
and for completeness set $\max(\varnothing) = 0$.
\end{definition}

\section{Classification of counting ratios} \label{S:TYPES_OP_PROFILE} 
A relative counting ratio is the conditional probability that a configuration in a reactive actuated automaton's trajectory coincides with a particular 
member of the marked set, given that it agrees with the marked set. We consider one other: an absolute counting ratio is 
the time rate at which a particular configuration of a trajectory coincides with any member of the marked set.

\subsection{Ratios of counting functions} \label{S:COUNTING_RATIOS} 
We apply the counting procedure of \S\ref{S:COUNTING}, which counts the occurrences 
$N_Z(\Idiom{\ConfigScript{c}}{n}, k)$ of marked set $Z$ in trajectory $\Idiom{\ConfigScript{c}}{n}$ 
before or at configuration $k$.

For a trajectory $\Idiom{\ConfigScript{c}}{n}$, given a marked set $Z$, a member $z$ of the 
marked set, and the whole space $\ConfigScript{C}$, there are two forms of counting function 
ratios of interest. First is the relative counting function ratio:
\[
q_z(k) = \frac{N_\Single{z}(\Idiom{\ConfigScript{c}}{n}, k)}{N_Z(\Idiom{\ConfigScript{c}}{n}, k)}
\]
There are $\Card{Z}$ possible relative counting ratios, one for each $z \in Z$. 

\[
Q(k) = \frac{N_Z(\Idiom{\ConfigScript{c}}{n}, k)}{N_{\ConfigScript{C}}(\Idiom{\ConfigScript{c}}{n}, k)} = \
       \frac{N_Z(\Idiom{\ConfigScript{c}}{n}, k)}{k}.
\]

Because $\ConfigScript{C}$ is the whole space, then 
$N_{\ConfigScript{C}}(\Idiom{\ConfigScript{c}}{n}, k)$ counts each configuration, and equals $k$.

The $q_z(k)$ are relative frequencies, whereas $Q(k)$ is absolute. The relative frequencies may be converted to absolute 
by multiplication: $Q_z(k) = q_z(k) \cdot Q(k)$.

\section{Limits of counting ratios} \label{S:LIMITS_OP_PROFILE} 

\subsection{Probability} \label{S:RELATIVE_OP_PROFILE} 
Let $\Idiom{\ConfigScript{c}}{n}$ be a trajectory. Suppose $z \in Z \subset \ConfigScript{C}$ is a configuration of the marked set. Software encounters 
$N_{\{z\}}(\Idiom{\ConfigScript{c}}{n}, k)$ instances of configurations satisfying $\lbrace \ConfigScript{c}_n \rbrace(i) = \ConfigScript{c}_i = z$ during the first 
$k$ automaton configurations. In the same execution there are $N_Z(\Idiom{\ConfigScript{c}}{n}, k)$ instances of $\ConfigScript{c}_i \in Z$. In the frequentist 
\cite{wW14fp} school of interpreting probability, when the limit exists, 
\[
P(z \mid Z) = \lim_{\;k \to \infty} \frac{q_z(k)}{Q(k)} 
            = \lim_{\;k \to \infty} \frac{N_{\{z\}}(\Idiom{\ConfigScript{c}}{n}, k)}{N_Z(\Idiom{\ConfigScript{c}}{n}, k)}
\] 
represents the conditional probability of occurrence of $z$, given that $Z$ occurs. 


\begin{definition}[relative counting ratio] \label{D:RELATIVE_OP}
Let $Z$ be a marked set of which $z$ is a member. A mapping $\OP$ is a \emph{relative} counting ratio over $Z$ if 
for all $z \in Z$, $(z, P(z \mid Z)) \in \OP$
\end{definition}

\begin{theorem}
A relative counting ratio is a finite probability distribution. 
\end{theorem}     
\begin{proof}
Definition \ref{D:MARKED_SET} asserts that marked set $Z$ is finite; therefore a mapping containing one abscissa per 
member of $Z$ is finite (definition \ref{D:RELATIVE_OP}). 
Since each ordinate is a ratio of counts, then each is nonnegative. 

The following shows that the sum of limits is unity:
\begin{alignat*}{3}
&\sum_{z \in Z} P(z \mid Z) &&= \sum_{z \in Z} \lim_{k \to \infty} 
                                \frac{N_{\{z\}}(\Idiom{\ConfigScript{c}}{n}, k)}{N_Z(\Idiom{\ConfigScript{c}}{n}, k)} \\
&                           &&= \lim_{k \to \infty} \sum_{z \in Z} 
                                \frac{N_{\{z\}}(\Idiom{\ConfigScript{c}}{n}, k)}{N_Z(\Idiom{\ConfigScript{c}}{n}, k)} 
                                &&\quad \textrm{[sum of limits equals the limit of sums]} \\
&                           &&= \lim_{k \to \infty} 
                                \frac{1}{N_Z(\Idiom{\ConfigScript{c}}{n}, k)} 
                                \sum_{z \in Z}{N_{\{z\}}(\Idiom{\ConfigScript{c}}{n}, k)} \\
&                           &&= \lim_{k \to \infty} 
                                \frac{N_Z     (\Idiom{\ConfigScript{c}}{n}, k)}{N_Z(\Idiom{\ConfigScript{c}}{n}, k)} 
                                &&\quad [N_Z(\Idiom{\ConfigScript{c}}{n}, k) = \sum_{z \in Z}{N_{\{z\}}(\Idiom{\ConfigScript{c}}{n}, k)}] \\
&                           &&= 1 
\end{alignat*}
Thus $P(z \mid Z)$ is a conditional probability distribution.
\end{proof}

A relative counting ratio is a mapping $\OP \colon Z \to [0,1]$ having total measure~1.

\subsection{Absolute ratio} \label{S:ABSOLUTE_OP_PROFILE} 
Let $\Idiom{\ConfigScript{c}}{n}$ be a trajectory.
An absolute counting ratio is the probability $P(Z)$ with which a trajectory (of some usage pattern) coincides with any 
configuration of the marked set $Z$. As before, this probability is the limiting ratio of two counting functions.
Its numerator contains $N_Z(\Idiom{\ConfigScript{c}}{n}, k)$, the same count as appears in the denominator of the relative counting ratio. 
In its denominator is the counting function of all possible configurations, namely $N_{\ConfigScript{C}}(\Idiom{\ConfigScript{c}}{n}, k)$, where $\ConfigScript{C}$ 
is the space of all configurations. 
Thus the ratio of counting functions of marked set $Z$ to the entire space $\ConfigScript{C}$
is $\frac{N_Z(\Idiom{\ConfigScript{c}}{n}, k)}{N_{\ConfigScript{C}}(\Idiom{\ConfigScript{c}}{n}, k)} = \frac{N_Z(\Idiom{\ConfigScript{c}}{n}, k)}{k}$, 
and the absolute counting ratio 
(of collection $Z$) is
\[
P(Z) = \lim_{\;k \to \infty} \frac{N_Z(\Idiom{\ConfigScript{c}}{n}, k)}{k}.
\]
\[
Q(Z) = \frac{N_Z(\Idiom{\ConfigScript{c}}{n}, k)}{k}.
\]

\chapter{Operational profiles} \label{S:OPERATIONAL_PROFILE_CHAPTER} 
For some trajectories the configuration counting ratio has an analytic limit. 
In this case the limiting quotient is called an operational profile. 

\begin{remark}
This chapter frequently uses the compound idiom that $\Idiom{x}{n}$ represents an anonymous 
sequence of objects of the same type as $x$. That is, if $X$ is the set of all $x_i$, then 
$\Idiom{x}{n} \colon \naturalnumbers \to X$.
\end{remark}

\section{Musa's operational profile}\label{S:MUSA_OP_PROFILE} 
Musa \emph{et al.} intended operational profiles as a tool for analysis of software reliability. A notion of the operational profile appeared 
in their pioneering exposition \cite{jM87}. The authors represent a program's abstract purpose as a collection of executable \Dquo{run types}
\negthickspace, which weren't discussed further. Musa posits that an operational profile is the program's set of run types along their 
probability of occurrence. 

Musa identified two intertwined concepts, the operational profile and the run type. The run type is envisioned as an indivisible unit 
associated with a probability. 

\section{Extension of Musa concept} \label{S:EXTENDED_OP_PROFILE} 
We will follow Musa's successful lead with these two concepts, except we will pursue a partition with finer granularity than the gross-scale 
run type. They will be partitioned into configurations, the elementary quantum of automata. This approach focuses on algorithmic structure, detaching 
the operational profile concept from higher-level human cognition of purpose. Despite appearances, this extension is not daunting~-- the only 
needed addition is a method for counting configuration events.

The synchronization function allows expression of the absolute operational profile as an intensity (rate or quasi-frequency). 

\begin{definition}
The \emph{temporal norm} is written using the double bar notation $\Norm{\cdot}$: 
\[
\Norm{Z} = \lim_{\;k \to \infty} \frac{N_Z(\Idiom{\ConfigScript{c}}{n}, k)}{\sync(\Idiom{\ConfigScript{c}}{n}, k)}.
\]
\end{definition}

The absolute operational profile is properly a subadditive seminorm on sets of configurations. As the limiting ratio of two counts in the 
natural numbers, the norm is positive. The norm is a seminorm because for some nonempty set $Z$ it may be true that $\Norm{Z} = 0$ 
(if the usage pattern does not activate any member of the marked set). This norm is subadditive because for any other set $S$, 
$N_{Z \cup S}(\Idiom{\ConfigScript{c}}{n}, k) \leq N_Z(\Idiom{\ConfigScript{c}}{n}, k) + N_S(\Idiom{\ConfigScript{c}}{n}, k)$. It follows that 
$\Norm{Z \cup S} \leq \Norm{Z} + \Norm{S}$. 

\section{Usage} \label{S:USAGE_SECTION} 
While the notion of \Dquo{purpose} will seem obvious to engineers, the same is less natural for mathematicians. Related to 
purpose is \Dquo{usage}. It is easier to explore through structure, which mathematically characterizes usage. The idea 
behind a usage is exercise of a small number of subordinate routines to demonstrate fitness for purpose. Since the number 
of routines comprising a usage is small, testing it repetitively will likely involve unnatural manipulation of stochastic 
stimulus. It is usually not good to define usage through behavior rather than philosophy.

\subsection{Usage colorings} \label{S:USAGE_COLORING} 
\begin{definition}[usage]
A \emph{usage} is an alternating coloring scheme for trajectories that partitions a trajectory into two sets of blocks. 
Each block of the partition is finite and colored either \emph{in} or \emph{out}.
\end{definition}

\begin{remark}
In-blocks share the common purpose of the usage; out-blocks share nothing. 
\end{remark}

\begin{definition}
A trajectory \emph{exercises} a usage if the in-blocks of the usage's partition appear infinitely often.
\end{definition}

\subsection{Usage colorings and marked sets} \label{S:USAGE_MARKED} 
\begin{definition}
Let $B$ be an in-block of the partition of trajectory $\tau$ induced by usage coloring $\#(\tau)$.  
Let $z = \inf{B}$ denote the \emph{infimum} (greatest lower bound) of $B$. 
\end{definition}

\begin{theorem}
$Z = \SetBuild{\inf{B}}{B \in \mathcal{B}_\mathrm{in}}$ is a marked set. 
\end{theorem}
\begin{proof}
Section \ref{S:COUNTING} places no restriction on marked sets other than that they be included in a trajectory.
$\SetBuild{\inf{B}}{B \in \mathcal{B}_\mathrm{in}}$ is a subset of a trajectory because the blocks consitute a partition 
of the trajectory, and the infimum is a member of a finite block..
\end{proof}

\begin{definition}
$Z = \SetBuild{\inf{B}}{B \in \mathcal{B}_\mathrm{in}}$ is called a \emph{usage-induced} marked set. 
\end{definition}

\section{Equivalence in usage} \label{S:USAGE_EQUIV} 
\begin{definition} \label{T:LIMIT_RATIO}
Two different trajectories $\Idiom{\ConfigScript{c}}{n}$ and 
$\IdiomPrime{\ConfigScript{c}}{n}$ are equivalent in usage if: 
\begin{alignat*}{3}
&&\lim_{k \to \infty} \frac
    {N_\Single{z}(\Idiom{\ConfigScript{c}}{n}, k)}
    {N_Z(\Idiom{\ConfigScript{c}}{n}, k)} 
        &= \lim_{k \to \infty} \frac
            {N_\Single{z}\IdiomPrime{\ConfigScript{c}}{n}, k)}
            {N_{Z}\IdiomPrime{\ConfigScript{c}}{n}, k)} 
        &&\quad \textrm{[relative profiles; one for each $z$ 
            in marked set $Z \neq \varnothing$]} \\
        &&\lim_{k \to \infty} \frac
            {N_{Z}(\Idiom{\ConfigScript{c}}{n}, k)}
            {N_{\ConfigScript{C}}(\Idiom{\ConfigScript{c}}{n}, k)}    
        &= \lim_{k \to \infty} \frac
            {N_{Z}(\IdiomPrime{\ConfigScript{c}}{n}, k)}
            {N_{\ConfigScript{C}} (\IdiomPrime{\ConfigScript{c}}{n}, k)} 
        &&\quad \textrm{
            [absolute profile; one for marked set $Z \neq \varnothing$].} 
\end{alignat*}
\end{definition}

\chapter{Indemnification} \label{S:INDEMNIFICATION_CHAPTER} 
Collections of related software tests are composed of fragments of a cone (chapter 
\ref{S:CONE_CHAPTER}). These tests individually pass or fail, but this perspective 
does not avail a collective view which considers that the collection of tests shares a 
common cone. 

Indemnification is a conglomerate statistic that considers the common cone relationship. 
Collections of tests related by a single cone are converted to equivalent error-free 
operational time. These durations are summed over the cone and interpreted as a Poisson 
process. This conversion enables software results to be expressed in traditional hardware 
assurance units. 


\section{Background} \label{S:INDEMNIFICATION_BACKGROUND} 

Hypothesize that a software hazard is emulated by a compound Poisson process (CPP) having 
intensity $\lambda$ and expected loss $\mu_L$.
Suppose further that the actual control mechanism is a cone convergent to the software point of exhibition of the hazard.  We wish 
to consider statistical evidence that the hazard's hypothetical description via the stochastic process is consistent with its 
mechanism as revealed by safety demonstration.
Indemnification is conversion of pass/fail test data into the failure intensity of an 
equivalent Poisson process.

\subsection{Conversion into a rate} \label{S:ABSOLUTE_OP_PROFILE_RATE} 
For each configuration of a trajectory, an amount of real time appropriate for an 
software system emulating the automaton's configuration is added to the time consumption 
budget. Let $Z$ be the usual arbitrary marked collection of configurations and $\Idiom{\ConfigScript{c}}{n}$ be a trajectory. These two 
provide a set of stochastic excitations and a sequence of configurations in which to count the events' arrivals. The synchronization records 
discrete pairs $(i, t_i)$, where $i$ is the index of the automaton configuration and $t_k$ is the total elapsed time after $k$ configurations. Call 
this mapping the synchronization function, having the formalism 
$\sync \colon \ConfigScript{C}^\naturalnumbers \times \naturalnumbers \to \positivereals$, along with assumed starting point 
$\sync(\Idiom{\ConfigScript{c}}{n}, 0) = 0$. 

Let the sequence index of each configuration be the discrete analog of time. Of course, this 
has the effect that discrete software time will not hold proportional to hardware real time. 
The approximate real time required by execution of configuration 
$\ConfigScript{c} = (\lambda, \Ftylc{f}, \Frm{f})$ is $\tau(\Ftylc{f})$ -- that is, elapsed real time is taken as a function of the executing functionality. 

\begin{definition} \label{D:SYNC_APPOX}
For trajectory $\Idiom{\ConfigScript{c}}{n}$, approximate time elapsed during the first $k$ configurations accumulates to 
\[
\sync(\Idiom{\ConfigScript{c}}{n}, k) = \sum_{i=1}^k \tau(\Ftylc{f}_i)) = \sum_{i=1}^k \tau(\mho_{\Fty{F}}(\ConfigScript{c}_i)) = t_k.
\]
\end{definition}

A theorem to avoid creating dependency on specific trajectories is in order. 

\begin{definition}
The \emph{temporal norm} with respect to trajectory $\Idiom{\ConfigScript{c}}{n}$ 
is written using double bar notation $\Norm{\cdot}_{\Idiom{\ConfigScript{c}}{n}}$: 
\[
\Norm{Z}_{\Idiom{\ConfigScript{c}}{n}} = 
    \lim_{\;k \to \infty} 
        \frac{N_Z(\Idiom{\ConfigScript{c}}{n}, k)}
        {\sync(\Idiom{\ConfigScript{c}}{n}, k)}.
\]
\end{definition}

The absolute counting ratio is properly a subadditive seminorm on sets of configurations. As the limiting ratio of two counts in the 
natural numbers, the norm is positive. The norm is a seminorm because for some nonempty set $Z$ it may be true that $\Norm{Z} = 0$ 
(if the usage pattern does not activate any member of the marked set). This norm is subadditive because for any other set $S$, 
$N_{Z \cup S}(\Idiom{\ConfigScript{c}}{n}, k) \leq N_Z(\Idiom{\ConfigScript{c}}{n}, k) + N_S(\Idiom{\ConfigScript{c}}{n}, k)$. It follows that 
$\Norm{Z \cup S} \leq \Norm{Z} + \Norm{S}$. 

\begin{remark}
The synchronization function allows expression of the absolute counting ratio as an intensity (rate or quasi-frequency). 
\end{remark}

\begin{definition} \label{D:TIME_RATIO}
For a marked set of configurations $\varnothing \neq Z \subseteq \ConfigScript{C}$, and 
different trajectories $\Idiom{\ConfigScript{c}}{n}$ and $\{\ConfigScript{c}'_n\}$ having the 
same temporal norm 
($\Norm{Z}_{\Idiom{\ConfigScript{c}}{n}} =  \Norm{Z}_{\{\ConfigScript{c}'_n\}}$), 
the trajectories are said to be of equivalent \emph{intensity} for the marked set.
\end{definition}

\chapter{Applied statistics} \label{S:STATISTICS_CHAPTER} 

The safety demonstration furnishes data for the indemnification statistic, which originates in the compound Poisson random process.

\section{Reliability demonstration} 
A reliability demonstration is a structured random experiment carrying controlled statistical uncertainty and providing \Dquo{hard} 
evidence against potential liability. 

\subsection{Safety demonstration} \label{S:SAFETY_DEMONSRATION} 
In software safety analysis, a hazard is a region of code bearing potential harmful side effects if incorrectly implemented. A safety 
demonstration is a special type of reliability demonstration posed to exercise a hazard. Here the region is presumed to be an acyclic 
cone, with the hazard located at its crux. The crux is a point of software/hardware transduction, illustrating the 
principle of emergence\footnote{software causes no harm until erroneous values transduce the boundary between software and hardware.} 
(see \S\ref{S:PRINCIPLE_OF_EMERGENCE}).

To oversimplify, a safety demonstration is a random sample from such a region (acyclic cone). The complete story is not so simple, 
because the cone is not a probabilistic structure; it possesses no probability to support randomness. 

As a probabilistic structure, the operational profile (\S\ref{S:RELATIVE_OP_PROFILE}) permits random sampling from its marked set, 
regardless of its higher level meaning. As the edge of a cone is a set, it can become a relative operational profile's marked set. 
Thus we tie an operational profile to a cone's edge. Let $\mathcal{O} \colon \edge{\Con{C}} \to [0,1]$ be a relative operational 
profile on the edge of cone $\Con{C}$. At this stage we have the ability to draw a random sample from $\edge{\Con{C}}$.

Theorem \ref{T:ACYLIC_CONE_CORRESPONDENCE} asserts that an acyclic cone $\Con{C}$ and $\edge{\Con{C}}$ are in one-to-one 
correspondence via the edge configuration relation of a bounded predecessor trajectory. Equivalent to the one-to-one correspondence is the bijection 
$\mathbf{b} = \{(\edge{\Stplc{w}},\Stplc{w}) \colon \Stplc{w} \in \Con{C}\}$. For $\ConfigScript{e} \in \edge{\Con{C}}$, $\mathbf{b}(\ConfigScript{e})$ is the 
bijectively corresponding bounded predecessor trajectory.

We now bijectively associate the random edge event $\ConfigScript{e} = \mathbf{b}^{-1}(\Stplc{w})$ with the bounded predecessor trajectory $\Stplc{w}$:
$\mathcal{O}' = \{ (\mathcal{O}(\mathbf{b}^{-1}(\Stplc{w})), \Stplc{w}) \colon \Stplc{w} \in \Con{C} \}$. With probability inherited from 
an operational profile, we can speak validly of a random sample from a cone.

\subsection{Tests} \label{S:TESTS} 
The last piece of the safety demonstration story is converting bounded predecessor trajectories into tests. Bounded predecessor trajectories are finite 
trajectories existing in confusion-prone backwards time. One may skip this section unless he wishes the detail of converting backward to forward 
trajectories.  

The \emph{test} function reverses and re-indexes bounded predecessor trajectories into conventional sequences. 

\begin{definition} \label{D:TEST}
Let $\Stplc{w}$ be a bounded predecessor trajectory of $n = \Card{\Stplc{w}}$ configurations, indexed from $0$ down to $-(n-1)$. 
Define the \emph{test} function $\test(\Stplc{w})  = \widetilde{\Stplc{w}}$ 
according to formula $\widetilde{\Stplc{w}}_i = \Stplc{w}_{i - n}$ for $i = 1, 2, \cdots , n$.
\end{definition}

Assuming that a bounded predecessor trajectory is indexed from $0$ down to $-(n-1)$,  
its corresponding test will be indexed from $1$ to $n$. 
In sense of direction, the bounded predecessor trajectory traverses configurations from  
$\ConfigScript{c}_{\text{crux}}$ to $\ConfigScript{c}_{\text{edge}}$, while the corresponding test traverses configurations from $\ConfigScript{c}_{\text{edge}}$ 
to $\ConfigScript{c}_{\text{crux}}$.

\begin{theorem}
Suppose $\Con{C}$ is an acyclic cone and $\Stp{W} \subset \Con{C}$ is a (unique) set of bounded predecessor trajectories. 
If \;$\widetilde{\Stp{W}} = \test(\Stp{W})$ is its converted set of reversed and re-indexed tests, then $\Stp{W}$ and $\widetilde{\Stp{W}}$ 
are in one-to-one correspondence. 
\end{theorem}
\begin{proof} 
By virtue of construction, $\test$ is already a mapping. Remaining to show is that $\test$ is additionally a bijection. Let $\ConfigScript{u}$ and 
$\ConfigScript{v}$ be bounded predecessor trajectories and $\ConfigScript{x}$ be a finite trajectory. As hypothesis set $\ConfigScript{x} = \test(\ConfigScript{u})  = \test(\ConfigScript{v})$. 
These sequences cannot be equal unless they possess the same number of terms, 
$n = \Card{\ConfigScript{x}} = \Card{\test(\ConfigScript{u})} = \Card{\test(\ConfigScript{v})}$. Since transformation $\test$ preserves the number of configurations (from 
Definition \ref{D:TEST} $\Card{\Stplc{w}} = \Card{\test(\Stplc{w})}$), then $n = \Card{\ConfigScript{x}} = \Card{\ConfigScript{u}} = \Card{\ConfigScript{v}}$. 

Again invoking Definition \ref{D:TEST} on the first part of the hypothesis, we write $\ConfigScript{x}_i = \ConfigScript{u}_{i - n}$. The second part 
similarly yields $\ConfigScript{x}_i = \ConfigScript{v}_{i - n}$. By equating the two parts, we now have $\ConfigScript{u}_{i - n} = \ConfigScript{v}_{i - n}$ for each $i$. 
In other words, the two bounded predecessor trajectories are actually the same trajectory: $\ConfigScript{u} = \ConfigScript{v}$. Thus $\test$ is a bijection. 
\end{proof}

\subsubsection{Stochastic variables preservation} \label{S:STOCHASTIC_VARIABLES} 
The danger in reversed thinking about tests is inadvertently conceptualizing stochastic variables as free. This is untrue, as the 
stochastic variables at any stage of a predecessor chain are fixed, and the \Dquo{next} stage considers the set of what previous conditions 
may have led to the current stage. Thus, predecessor trajectories are chains of a poset of configurations, which include the settings of stochastic variables. 
One must be mindful to reproduce all stochastic stimuli of the bounded predecessor trajectory in its analogous test.

\subsection{Outcome} \label{S:OUTCOME} 
The outcome of a test, pass or fail, will be regarded as a Bernoulli random event, $P_\rho = \rho^n {(1 - \rho)}^{1 - n}$,
for $n = 1$ (pass) or $n = 0$ (fail). These probabilities are statistically independent of the bias involved with drawing the sample
from the operational profile. This bias affects the origin of discovered failures, but not how many failures are found. In other words, 
the total statistical power of the sampling plan is not affected by sampling bias.

Sums of independent Bernoulli random variables are binomial. That is, the probability of finding $n$ failures collectively among 
$N$ sample items is binomial, $\binom{N}{n} \rho^n {(1 - \rho)}^{N - n}$.

\subsection{Physics} \label{S:PHYSICS} 
In the real world, tests pass or fail depending on whether the information transduced at configuration $\ConfigScript{c}_\text{crux}$ meets all safety 
constraints. Such engineering requirements are varied, ultimately involving position, timing, voltage, insulation, dimensional tolerance, 
toxicity, temperature, mechanical shielding, luminosity, and hydrostatic pressure -- just to name a few areas. Review of a test offers a 
last chance to discover a missed constraint (requirement). Another possibility is that the chain of precursor events should actually lead 
to a different conclusion.

Transduced values potentially control the status of any safety concern.  Tests simply pass or fail, but evaluation of why a test passes 
or fails can become nontrivial, requiring collaboration between mechanical, software, and system safety engineers. 

\subsection{Statistics} \label{S:STATISTICS} 
Some statistical error originates in inference from random sample to \Dquo{unknown} population (parametric family of probability 
distributions on a measurable space). Just one distribution is true, while the others are false. An assertion separating the 
parameterization into two decision units is called a hypothesis. One decision unit is traditionally designated null, while the 
other is called alternate. The true distribution belongs either to the null or alternative decision units.

Each sample item either passes or fails its associated test (see \S\ref{S:PHYSICS}). Within the entire cone $\Con{C}$, suppose the 
proportion of tests that fail is $\rho$. This proportion is subsequently realized approximately through a random sample. Regardless 
of the sample size, since the application is to safety, the only cases of interest will be when the number of failures is zero. 
Other cases, implying need for reliability growth, are 
treated in the literature, particularly \cite{jM87}.

We now examine the case defined by drawing a random sample of size $N$ from $\edge{\Con{C}}$ and allowing $n = 0$ failures
in the associated tests from cone $\Con{C}$. The null decision unit contains the probability distribution
$P_0(\text{pass}) = 1$ and $P_0(\text{fail}) = 0$. The alternate decision unit is the set of probability distributions $P_\rho$ 
having $0 < \rho \leq 1$. Hypothesis evaluation entails two types of error, known as $\alpha$ and $\beta$ error. 

\subsubsection{False rejection ($\alpha$ error)} \label{S:FALSE_REJECTION} 
The first is false rejection of the null decision unit, with associated measurement error $\alpha$.
The sampling plan can reject only if finds an error, so this sampling plan is incapable of false rejection. Thus $\alpha \equiv 0$.

\subsubsection{False acceptance ($\beta$ error)} \label{S:FALSE_ACCEPTANCE} 
The second is false acceptance of the null decision unit, with associated measurement error $\beta$. We experience false acceptance 
when $0 < \rho$ but the sample contains no failures.

Under the binomial model, the probability of observing a random sample of size $N$ with $n$ failures collectively is 
$\binom{N}{n} \rho^n {(1 - \rho)}^{N - n}$. Proceeding to our case of interest, $n = 0$, we have 
$\binom{N}{n} \rho^n {(1 - \rho)}^{N - n} \bigm|_{n=0} = {(1 - \rho)}^{N}$. This expression is the probability that random samples of 
size $N$ from a source of characteristic failure proportion $\rho$ will be accepted. 

\subsubsection{Power function} \label{S:POWER_FUNCTION} 
It is confusing to reason in terms of contravariant\footnote{One increasing, the other decreasing} attributes. In our case we formulate  probability of rejection as an increasing function of $\rho$, a measure of the population's undesirability. The probability that random 
samples will be properly rejected is the previous expression's complement: 
\[
K_{N,0}(\rho) = 1 - {(1 - \rho)}^{N} = 1 - \beta.
\]

This non-contravariant result is known as the power function of sample size $N$, tolerating zero $(0)$ failures.
The graph of the power function always increases, starting at $0$ for $\rho = 0$ and ending at $1$ for $\rho = 1$.
Just how fast this function increases in its midrange is determined by the sample size $N$.
With sample size one $(N = 1)$, $K_{1,0}(\rho) = \rho$.

\begin{table}[H]
    \begin{center}
        \begin{tabular}{|r|c|c|c|c|c|c|}
            \hline
$N$&$K_{N,0}(.001)$&$K_{N,0}(.01)$&$K_{N,0}(.05)$&$K_{N,0}(.10)$&$K_{N,0}(.50)$&$K_{N,0}(.90)$ \\ \hline
1&.0010&.0100&.0500&.1000&.5000&.9000 \\
5&.0050&.0490&.2262&.4095&.9688&1.0000 \\
10&.0100&.0956&.4013&.6513&.9990&1.0000 \\
15&.0149&.1399&.5367&.7941&1.0000&1.0000 \\
20&.0198&.1821&.6415&.8784&1.0000&1.0000 \\
30&.0296&.2603&.7854&.9576&1.0000&1.0000 \\
50&.0488&.3950&.9231&.9948&1.0000&1.0000 \\
100&.0952&.6340&.9941&1.0000&1.0000&1.0000 \\
200&.1814&.8660&1.0000&1.0000&1.0000&1.0000 \\
500&.3936&.9934&1.0000&1.0000&1.0000&1.0000 \\
1000&.6323&1.0000&1.0000&1.0000&1.0000&1.0000 \\
2000&.8648&1.0000&1.0000&1.0000&1.0000&1.0000 \\
5000&.9933&1.0000&1.0000&1.0000&1.0000&1.0000 \\
10000&1.0000&1.0000&1.0000&1.0000&1.0000&1.0000 \\ \hline
        \end{tabular}
    \end{center}
    \caption{Family of power functions (probability of rejection)} \label{Ta:POWER_FUNCTION}
\end{table}

Within this family $\beta = {(1 - \rho)}^{N} = 1 - K_{N,0}(\rho)$. 

One is initially dismayed by this sketch of the family of power functions; it suggests that high degrees of assurance are unobtainable 
through random sampling using practical sample sizes. However, reasonable performance useful for coarser screening is very possible. 
Detecting a defective population of 10 percent with a probability of approximately 90\% requires only 20 sample items. 

\subsection{Sampling philosophy} \label{S:SAMPLING_PHILOSOPHY} 
Our safety demonstration sampling technique contrasts two assurance philosophies -- software reliability versus software correctness. 
The software reliability perspective involves a separate operational profile on $\edge{\Con{C}}$, whereas software correctness examines 
only the structure within cone $\Con{C}$. The operational profile asserts the importance of relative excitational intensity to safety 
analysis. An accident that occurs more frequently is worse than an accident that happens less frequently, given that they are of 
comparable severity. This safety factor is ignored under software correctness alone.

\section{Modeling of accidents} 
Accidents are diverse in effect and mechanism, including injury, death, or damage either to equipment or environment. Since the 
causality of accidents is temporarily unknown, they manifest an apparent nature of unpredictability or randomness. However, under 
emulation as a stochastic process, the exact timing of accidents \emph{is} truly a random phenomenon rather than causal. Nevertheless, 
it has proven useful to compare well-understood summary statistics of stochastic processes with those of deterministic but unknown 
physical processes.

\subsection{Compound Poisson process} \label{S:CPP} 
Today's prevalent safety model for the occurrence of accidents is the compound Poisson\footnote{After Sim\'{e}on Denis Poisson, 
mathematician and physicist, 1781 -- 1840} process. This model captures accidents' two dominant attributes: rate of occurrence (intensity) 
and scalar measure of loss (severity). With some exceptions, neither the timing nor severity of one software accident affects another. 
The compound Poisson process (CPP) is appropriate to model accidents of this nature. 

As stochastic processes are models rather than mechanisms, deriving their properties involves somewhat out-of-scope mathematics.
The interested reader can immediately find greater detail in Wikipedia\textregistered \ online articles:
\cite{wW_Poisson_distribution}, \cite{wW_Poisson_process}, \cite{wW_total_expectation}, \cite{wW_Compound_Poisson_process}, 
and \cite{wW_Cumulant}. Relevant theorems will be documented here simply as facts.

\subsection{Poisson processes}\label{S:POISSON_PROCESSES} 
We will consider three variants of basic stochastic process: the ordinary Poisson process, the compound Poisson process,
and the intermittent compound Poisson process.

\subsubsection{Ordinary Poisson process}\label{S:INTERMITTENT_POISSON} 
(Ordinary) Poisson processes are characterized simply by their rate or intensity:
\begin{itemize}
\item its fundamental rate $\lambda$, which is the expected number of arrivals per unit time. 
\end{itemize}

\begin{fact}
Let $\lambda$ be the rate of a Poisson process. The probability of experiencing $k$ arrivals in a time interval 
$t$ units long is
\[
P_\lambda(k) = e^{-\lambda t}\frac{{(\lambda t)}^k}{k!}.
\]
\end{fact}

\subsubsection{Compound Poisson process}\label{S:INTERMITTENT_POISSON} 
A compound Poisson process is characterized by two rates:
\begin{itemize}
\item its fundamental rate $\lambda$ as before, and 
\item its rate of loss $L$, which is a random variable invoked once for each arrival. 
\end{itemize}

\begin{fact}
Let $\lambda$ be the rate and $L$ be the loss random variable of a compound Poisson process. The expectation of the compound 
process for a time interval $t$ units long is
\begin{alignat*}{1}
\E (\text{compound Poisson}) &= \lambda t \cdot \E (L)\\
                             &= \lambda t \cdot \mu_L.
\end{alignat*}
\end{fact}

\begin{definition} \label{D:STATISTICAL_RISK}
The statistical \emph{risk}, written $h$, of a compound Poisson process is the time derivative of its expectation in a duration
of length $t$; that is
\[
h = \frac{d}{dt} \E(\text{compound Poisson}) =  \frac{d}{dt} (\lambda t \cdot \mu_L) = \lambda \mu_L,
\]
which is the product of its rate $\lambda$ and its expected loss $\mu_L$.
\end{definition}

\subsubsection{Intermittent compound Poisson process}\label{S:INTERMITTENT_POISSON} 
A variation of the CPP is the intermittent compound Poisson process, which is intermittently on or off with expected durations
$\E(\text{on}) = \mu_\text{on}$ and $\E(\text{off}) = \mu_\text{off}$. 
An intermittent compound Poisson process (ICPP) is characterized by three rates:
\begin{itemize}
\item its fundamental rate $\lambda$ as before, and 
\item its rate of loss $L$, also as before,
\item alternating durations of random lengths
      $\tau_\text{on}$ and $\tau_\text{off}$.
\end{itemize}

Random variables $\tau_\text{on}$ and $\tau_\text{off}$ converge to $\mu_\text{on}$ and $\mu_\text{off}$ in the limit.
The \emph{idle} ratio of a intermittent compound Poisson process is 
$\iota = \frac{\mu_\text{off}}{\mu_\text{on} + \mu_\text{off}}$.

\begin{fact}
Let $\lambda$ be the rate, $L$ be the loss random variable, and $\iota$ be the idle ratio of an intermittent compound Poisson process. 
The expectation of the ICPP for a time interval $t$ units long is
\begin{alignat*}{1}
\E (\text{intermittent compound Poisson}) &= (1 - \iota) \cdot \lambda t \cdot \E(L) \\
                                          &= (1 - \iota) \cdot \lambda t \cdot \mu_L.
\end{alignat*}
\end{fact}

The statistical risk of an ICPP is 
\begin{alignat*}{1}
h &= \frac{d}{dt} \E(\text{intermittent compound Poisson}) \\
  &= \frac{d}{dt} ((1 - \iota) \cdot \lambda t \cdot \mu_L + \iota \cdot 0 t  \cdot 0) \\
  &= (1 - \iota) \lambda \mu_L.
\end{alignat*}

\section{Indemnification} \label{S:INDEMNIFICATION_FORMULA} 
Hypothesize that a software hazard is emulated by a compound Poisson process (CPP) having 
intensity $\lambda$ and expected loss $\mu_L$.
Suppose further that the actual control mechanism is a cone convergent to the software point of exhibition of the hazard.  We wish 
to consider statistical evidence that the hazard's hypothetical description via the stochastic process is consistent with its 
mechanism as revealed by safety demonstration.

\subsection{Unification}\label{S:UNIFICATION} 
Before undertaking the question of whether test data supports a hypothetical stochastic process, we must establish the theoretical 
conditions under which equality is expected.

\subsubsection{Fundaments of the model} 
The compound Poisson process is a model stochastic process for occurrence of accidents. 
This model is used in safety analysis to quantify the occurrence and losses of accidents without considering their causes. 
\MilStd{} (see Appendix \ref{S:MIL-STD-882}) is an important example. 
In a time interval of duration $t$, accidents converge stochastically in rate to expectation $\lambda t$ and in mean loss to $\mu_L$.
This means an intensity of $\lambda$ accidents per time unit.

\subsubsection{Fundaments of the mechanism} 
The reactive actuated automaton is a mechanism representing software. When extended by the principle of emergence (\S\ref{S:PRINCIPLE_OF_EMERGENCE}) 
and the constructs of the operational profiles (\S\ref{S:OPERATIONAL_PROFILE_CHAPTER}) and cones (\S\ref{S:CONE_SECTION}), 
it becomes capable of representing precursor conditions for software accidents.
Let $\Norm{\edge{\Con{C}}}$ (see \S\ref{S:ABSOLUTE_OP_PROFILE_RATE}) be the rate-based absolute operational profile of the edge 
of an acyclic cone $\Con{C}$.  
Since a member of $\edge{\Con{C}}$ is executed at the average intensity of ${\Norm{\edge{\Con{C}}}}$, then so is 
the cone's configuration of convergence $\ConfigScript{c}_\text{crux}$. 
Let $\rho$ be the proportion of failing tests (bounded predecessor trajectories). 
Under that supposition, failures occur at the intensity of $\rho \cdot \Norm{\edge{\Con{C}}}$. 
The definition of $\Norm{\edge{\Con{C}}}$, 
through the internal function $\text{sync}(\cdot)$, allows for the passage of time in the proper duration.

\subsubsection{Uniting mechanism and model} 
We presume that one failing test equals one accident.
The cone's configuration of convergence is considered to be the point of exhibition of a hazard whenever safety constraints are not met.
This mechanism may be separately equated to the intensity (not the rate of loss) of the compound Poisson process: 
\[
\lambda = \rho \cdot \Norm{\edge{\Con{C}}}.
\]

This equation places a property of the model on the left and properties of the mechanism on the right. 

\begin{remark}
The execution rate of the edge of an acyclic cone numerically equals the execution rate of the (set containing the) cone's crux. 
The cone's definitional status (as a complete independent set of bounded predecessor trajectories ending at $\ConfigScript{c}_{\text{crux}}$) 
causes this. Symbolically,
\[
\Norm{\edge{\Con{C}}} = \Norm{\lbrace\ConfigScript{c}_{\text{crux}}\rbrace}.
\]
\end{remark}

\subsection{Evidence}\label{S:EVIDENCE} 
We propose that the same data used earlier for indemnification testing be re-used in a slightly different statistical context.
Recall that an indemnification test has $N$ items among which are zero failures, where each item is a bounded predecessor trajectory, 
and a cone is a structured collection of bounded predecessor trajectories of an automaton.

We wish to measure the amount of information in an indemnification sample to explain the phenomenon that larger samples justify
more precise estimates than smaller samples. We refer to this information as measuring the \emph{weight of evidence}. 
%
This situation differs from the familiar problem of finding the maximum likelihood estimator.

\subsubsection{Method of indifference} \label{S:INDIFFERENCE_POWER_FUNCTION} 
The power function of sample size $N$, tolerating zero $(0)$ failures, is $K_{N,0}(\rho) = 1 - {(1 - \rho)}^{N}$ 
(see \S\ref{S:POWER_FUNCTION}). 
It measures the probability of rejection as a function of $\rho$.

Each power function $K_{N,0}(\rho) = 1 - {(1 - \rho)}^{N}$ is characterized by its indifference proportion, which is defined as the 
proportion at which rejection and acceptance become equally likely (that is, 
$K_{N,0}(\hat{\rho}_\text{\,I}) = {^1\!\!/\!_2} = 1 - K_{N,0}(\hat{\rho}_\text{\,I})$). With only modest algebra, the analytic expression for the 
indifference proportion may be derived from the power function $K_{N,0}(\rho)$; it is 
\[
\hat{\rho}_\text{\,I} = 1 - \sqrt[N]{{^1\!\!/\!_2}}. 
\]

Below is a numerical tabulation of the previous formula:

\begin{table}[h!]
    \begin{center}
        \begin{tabular}{|r|c|}
            \hline
$N$&$\hat{\rho}_{\,\text{indifference}}$ \\ \hline
1&.50000 \\
5&.12945 \\
10&.06697 \\
15&.04516 \\
20&.03406 \\
30&.02284 \\
50&.01377 \\
100&.00691 \\
200&.00346 \\
500&.00139 \\
1000&.00069 \\
2000&.00035 \\
5000&.00014 \\
10000&.00007 \\ \hline
        \end{tabular}
    \end{center}
    \caption{Indifference proportion} \label{Ta:INDIFFERENCE_PROPORTION}
\end{table}

\subsubsection{Indemnification formula} \label{S:UPPER_BOUND} 
The indemnification formula provides a statistical upper bound on hazard intensity. Indemnification data may be expressed as an 
equivalent statistical upper bound on hazard intensity. This differs fundamentally from estimating the intensity of a hazard. By 
a statistically \Dquo{guaranteed} hazard intensity, we mean an upper bound such that the true hazard intensity is likely to fall 
beneath this level with known confidence (probability).

Suppose we choose ${^1\!\!/\!_2}$ as the known confidence. The indifference proportion 
$\hat{\rho}_\text{\,I} = 1 - \sqrt[N]{{^1\!\!/\!_2}}$ then has a second interpretation as an upper bound with confidence 
${^1\!\!/\!_2}$. For any $\rho \leq \hat{\rho}_\text{\,I}$, it is true that power function $P_{N,0}(\rho) \leq {^1\!\!/\!_2}$, 
so $\hat{\rho}_\text{\,I}$ is an upper bound of confidence ${^1\!\!/\!_2}$.

To convert from the size of the indemnification sample into its equivalent upper bound hazard intensity, find the indifference 
proportion $\hat{\rho}_{\,\text{I}} = 1 - \sqrt[N]{{^1\!\!/\!_2}}$.

Check the sample physics. This amounts to analysis of the originating cone $\Con{C}$, which is the point of exhibition of a hazard 
whenever safety constraints are not met. The cone's edge has an absolute operational profile expressed as a rate. This quantity is 
the temporal norm of the cone's edge.

We have shown that the probable upper bound of the hazard intensity is proportional to the indifference proportion, with constant of 
proportionality furnished by the temporal norm of the cone's edge. The indemnification formula is:
\[
\hat{\lambda}_{\,\text{I}} = \hat{\rho}_{\,\text{I}} \cdot \Norm{\edge{\Con{C}}} 
= \hat{\rho}_{\,\text{I}} \cdot \Norm{\lbrace\ConfigScript{c}_{\text{crux}}\rbrace}.
\]

\chapter{Epilogue} 

From previous discussion two structures of system safety emerge: the safety demonstration and indemnification, 
its measure of assurance. Opinion follows.

\section{Programmatic fit} \label{S:PROGRAMMATIC_FIT} 
Safety demonstration and indemnification merge smoothly into today's programmatic picture. Early in the development cycle, safety 
engineers provide \Dquo{ballpark} quantifications of the threat of hazards\cite{DD12}, expressed as intensity and severity. These 
numbers are often educated guesses: a mixture of circumstance, intuition, similar design, and history. At that stage, the process 
is without supporting evidence.  Later in the development cycle, assuming a program of structured testing has been followed, 
statistical evidence is available in the form of a safety demonstration. These data are expressed as a statistical upper bound on 
each software hazard's intensity -- that is, an indemnification -- and used as evidence of correct operation. This additional configuration 
frees the safety engineer from having to re-asses his original estimate using the same shallow method as the original guesstimate.

Increasingly, indigenous static syntax analyzers satisfy need for overall code robustness. 
However, exclusive reliance on these analyzers would result in a software safety engineering shortcoming, 
because they do not always detect code defects that have valid syntax (that is, syntactically valid but wrong algorithm).

\section{Commercialization} \label{S:COMMERCIALIZATION} 
Difficult work remains before safety demonstration and indemnification can be supported as mature commercial technology. The role 
of the reactive actuated automaton must be replaced by a real-world programming language. Present theory restricting tests (predecessor trajectories) 
to acyclic cones may require generalization to cyclic cones to achieve broader range. Safety demonstration demands the ability 
to produce approximate operational profiles from which can be drawn pseudo-random samples. Spin-offs from similar technologies may 
be possible; static analyzers are one example.

\section{Criticism of MIL-STD-882} \label{S:CRITICISM} 
Statistical risk is a model emulating the threat of haphazard accidents. Users of \MilStd{} are familiar 
with statistical risk as an accident model for hardware. Software is properly deterministic and therefore non-stochastic, but it's 
successfully approximated with the same Poisson stochastic process as hardware. The rationale for this approximation is to apply a 
useful Poisson mathematical assurance property; thus assurance depends on the validity of the Poisson process as an approximation 
to the reactive actuated automaton. 

The Standard introduces a risk-like scale replacing statistical risk for software. For purpose of this chapter, we call this replacement scale the 
\emph{design} risk. The discussion of chapters 2 and 3 conclude that, from the standpoint of statistical assurance, there is no justification for 
the presently differing versions of the term \emph{risk} between hardware and software. 

The state of software engineering is mixed science and sophisticated art. In the current Standard, art has somewhat overtaken science; for 
software, the concept of statistical risk has been abandoned. One subsequently loses the ability to measure and assure risk uniformly. 

Design risk and statistical risk do share a severity axis, but the similarity ends there. Statistical risk has another axis composed of a numerical 
product, the frequency of execution times the probability of error. What we have called design risk also possesses another axis, but it is a 
categorical scale arranged in decreasing order of the design safety importance of the software's functionality.

This results in an error in \MilStd{} with serious consequence. Because the frequency of execution of a software point is not 
well-correlated with its functionality's design safety importance, statistical risk does not correlate well with design risk. Because the Standard 
assigns statistical risk to hardware and design risk to software, and due to lack of correlation between the two, there is no way to rank the relative 
importance of hazards of mixed type. Loss of the ability to compare risks of all hazards is a flagrant omission. Under correct physics, risks of multiple 
hazards are additive. This is not the case under \MilStd.

The formal sense of \emph{assurance} is lost by these definitional variants. Being quantitatively assured requires a limit value on proportion or 
mean deviation and a statement of statistical confidence for this limit; the Standard clearly lacks this characteristic. Properly assurance is a 
numeric quantity associated with statistical control of risk, not an engineering activity to further psychological confidence. Despite that its 
developers may express great confidence in the methodology, software built under the Standard is not quantitatively assured.


\section{Repair of MIL-STD-882} \label{S:REPAIR} 
Rehabilitation of \MilStd{} is straightforward. It must be amended to contain an engineering introduction to statistical risk for software, 
including allied procedures. This subject matter is covered here in mathematical language, but should be presented differently for engineers' 
consumption. The revised Standard should distinguish between formal assurance and design confidence, and classify what procedures support either.
Generally, the concerns associated with design risk align with developmental software engineering, while those of statistical risk align with 
responsibilities of system safety engineering, part of systems engineering. Software and Systems Engineering should not duplicate each other's efforts.

\appendix

\chapter{Groundwork} \label{Ch:GROUNDWORK} 

This appendix examines ensembles, a mathematical structure in the theory of systems representing physical stimulus and 
response. From this start, discussion proceeds into the Cartesian product, choice spaces and subspaces, and partitions of choice 
spaces into deterministic and stochastic components. Dyadic notation is introduced.

\section{Ensemble} \label{S:ENSEMBLE} 
Using a single symbol, ensembles compactly represent physical stimuli, which are in turn associative mappings between \emph{entities} 
and \emph{values}. Here an entity is a unique descriptive category name such as \Dquo{leftmost rear light intensity control} and \emph{value} 
is the setting of the light at a particular instant. One must be cautious to realize engineering specificity in the naming of entities, 
as \Dquo{leftmost rear light intensity control} is meaningless when the light source has failed. A supplementary entity, for example 
\Dquo{leftmost rear light failure status,} may be necessary for engineering completeness. Physical systems are composed of interrelated 
combinations of stimulus and response, in which response is \emph{completely} determined by stimulus. An ensemble represents a system 
stimulus at a particular instant of time.

Mathematically, an ensemble is a special form of a more general structure known as a \emph{family}. We employ nomenclature abridged from 
Halmos \cite[p. 34]{pH74} as follows:
\begin{nomenclature}
Let $I$ and $X$ be non-empty sets, and $\varphi \colon I \to X$ be a mapping. Each element $i \in I$ is an \emph{index}, while $I$
itself is an \emph{index set}. The mapping $\varphi$ is a \emph{family}; its codomain $X$ is an \emph{indexed set}. An ordered pair
$(i,x)$ belonging to the family is a \emph{term}, whose value $x = \varphi(i)\in X$ is often denoted $\varphi_i$.
\end{nomenclature}
\begin{notation}
The family $\varphi$ itself is routinely but abusively denoted $\{\varphi_n\}$. This notation is a compound idiom.\footnote{
compound idiom: the symbol is a composite of other notational devices that do not together indicate the usual compound application 
of the symbol's separate sources; rather, the symbol has a meaning that must be recognized whole-wise.}
Especially in the case of sequences over a set $G$, the symbol $\{g_n\}$ signifies the mapping 
$\{ 1\mapsto g_1, \; 2\mapsto g_2, \; \ldots \;\}$.
\end{notation}

\begin{definition}\label{D:ENSEMBLE}
An \emph{ensemble} $\upsilon$ is a non-empty finite family (mapping).
\end{definition}

\begin{remark}
Since physical systems possess only a finite number of attributes, the scope of interest is ensembles having 
finite-dimensional index sets. 
\end{remark}
\begin{notation}
Let $\psi$ be an ensemble. For term $(i,v) \in \psi$, we denote $v = \psi_i$. 
\end{notation}

\section{Ensemble relations and operations} \label{S:ENSEMBLE_ARITHMETIC} 
This section presents binary comparisons and operators on ensembles. 

\begin{definition} \label{D:DISJOINT_ENSEMBLES}
Two ensembles $\psi$ and $\phi$ are \emph{disjoint} if $\domain{\psi} \cap \domain{\phi} = \varnothing$.
\end{definition}

\begin{definition} \label{D:ENSEMBLE_SUBSET}
Let $\psi$ and $\theta$ be ensembles. The notation $\theta \subseteq \psi$ signifies that the restriction 
$\Restrict{\psi}{\domain{\theta}} = \theta$.
\end{definition}

\begin{theorem} \label{T:UNION_ENSEMBLES}
The ordinary set-theoretic union of two disjoint ensembles $\psi$ and $\phi$ is another ensemble $\psi \cup \phi$.
\end{theorem}
\begin{proof}
Definition \ref{D:ENSEMBLE} asserts that ensembles $\psi$ and $\phi$ are non-empty finite families, which are non-empty mappings composed 
of a finite set of ordered pairs. The union of two finite sets of ordered pairs is another finite set of ordered pairs. This union is a 
mapping because definition \ref{D:DISJOINT_ENSEMBLES} assures the domains have no member in common. The mapping is an ensemble because it 
is non-empty and finite.
\end{proof}

\begin{remark}
The mathematical question is posed by $(\mathit{index}, \mathit{value}_1) \in \psi$ and $(\mathit{index}, \mathit{value}_2) \in \phi$, 
so just what unique $(\mathit{index}, \mathit{undefined})$ belongs to $\psi \cup \phi$? This represents a potential physical conflict 
of stimulus which we wish to avoid categorically by requiring disjointness.
\end{remark}

\begin{definition} \label{D:ENSEMBLE_SUBSET}
Let $\psi$ and $\theta$ be ensembles. The notation $\theta \subseteq \psi$ signifies that the restriction 
$\Restrict{\psi}{\domain{\theta}} = \theta$.
\end{definition}

\begin{definition}\label{D:ENSEMBLE_RANGE}
Let $\psi$ be an ensemble.
The \emph{range} of the ensemble  is
\[
\range{\psi} = \bigcup_{i \in \domain{\psi}} \psi(i).
\]
\end{definition}

\subsection{Ensemble union and dyadic notation} \label{S:ENSEMBLE_DYADIC_NOTATION} 
Dyadic notation is useful in single frequently-used applications which demand notational compactness. 
Dyadic notation replaces a binary operator invocation by its two concatenated arguments, thus resembling a product: 
\[
\mathit{operatorname}(\theta,\phi) \mapsto \theta\phi.
\]

\begin{definition} \label{D:DYADIC_ENSEMBLE_PRODUCT}
Let $\theta$ and $\phi$ be disjoint ensembles. By using \emph{dyadic} product notation 
\[
\theta \cup \phi = \theta \phi,
\] disjointness (definition \ref{D:DISJOINT_ENSEMBLES}) is implicitly implied.
\end{definition}

\begin{theorem}
Let $\theta$ and $\phi$ be disjoint ensembles; then $\theta\phi = \phi\theta$.
\end{theorem}
\begin{proof}
By twice applying definition \ref{D:DYADIC_ENSEMBLE_PRODUCT}, $\thickspace\theta \phi = \theta \cup \phi$ and $\phi \theta = \phi \cup \theta$.
But $\theta \cup \phi = \phi \cup \theta$, so $\theta\phi = \phi\theta$ by transitivity of equality.
\end{proof}

\section{Class} \label{S:CLASS} 
\begin{definition}\label{D:UNIFORM_SET_ENSEMBLES}
A domain-uniform collection of ensembles is a set $\Upsilon$ such that for any $\upsilon_1, \upsilon_2 \in \Upsilon$,
$\domain{\upsilon_1} = \domain{\upsilon_2}$.
\end{definition}

\begin{definition}\label{D:CLASS_ENSEMBLES} (equivocation)
The term \emph{class} is a synonym for a domain-uniform set of ensembles.
\end{definition}

\section{Class relations and operations} \label{S:CLASS_ARITHMETIC} 
\begin{definition}\label{D:CLASS_DOMAIN} 
Let $\Upsilon$ be a class. The idiom $\domain{\Upsilon}$ means $\domain{\upsilon}$, where $\upsilon \in \Upsilon$.
\end{definition}

\begin{lemma}\label{L:CLASS_DOMAIN} 
Definition \ref{D:CLASS_DOMAIN} is unique.
\end{lemma}
\begin{proof} 
Let $\Upsilon$ be a class and suppose $\upsilon_1, \upsilon_2 \in \Upsilon$. 
Since $\Upsilon$ is a class, then it is also a domain-uniform set of ensembles, and, by definition \ref{D:UNIFORM_SET_ENSEMBLES}$, 
\domain{\upsilon_1} = \domain{\upsilon_2}$.  
\end{proof}

\begin{remark} 
Although it has no proper domain, a class is a set of mappings which possesses a common domain.
\end{remark}

\begin{definition}\label{D:CONSTANT_CLASS}
Let $\Upsilon$ be a class. If for index $i \in \domain{\Upsilon}$ there exists constant $c_0$ such that $\upsilon(i) = c_0$ 
for any ensemble $\upsilon \in \Upsilon$, then the class $\Upsilon$ is \emph{constant} at $i$.
\end{definition}

In the context of classes, the symbol usually designating \Dquo{subset} instead indicates quasi-equality between two classes: 
\begin{definition} \label{D:CLASS_SUBSET}
Let $\Psi$ and $\Phi$ be classes. The notation $\Phi \subseteq \Psi$ means that 
for every $\psi \in \Psi$, there exists $\phi \in \Phi$ such that $\Restrict{\psi}{\domain{\Phi}} = \phi$,
\emph{and} 
for every $\phi \in \Phi$, there exists $\psi \in \Psi$ such that $\Restrict{\psi}{\domain{\Phi}} = \phi$. 
\end{definition}

\begin{definition} \label{D:CLASS_RANGE}
Let $\Psi$ be a class. 
The \emph{range} of $\Psi$ is the union the ranges of all constituent ensembles:
\[
\range{\Psi} = \bigcup_{\psi \in \Psi} \range{\psi}.
\]
\end{definition}

\subsection{Class union and dyadic notation} \label{S:CLASS_DYADIC_NOTATION} 
\begin{definition}\label{D:DISJOINT_CLASSES} 
Let $\Psi$ and $\Upsilon$ be classes. The classes are \emph{disjoint} if $\domain{\Psi} \cap \domain{\Upsilon} = \varnothing$ 
(using definition \ref{D:CLASS_DOMAIN}).
\end{definition}

Recall from definition \ref{D:ENSEMBLE_RANGE} that the range of an ensemble $\psi$ is 
$\range{\psi} = \cup_{i \in \domain{\psi}} \medspace\psi(i)$.

\begin{definition} \label{D:CLASS_UNION}
The \emph{union} of two disjoint classes $\Theta$ and $\Phi$ is
\[
\Theta  \cup \Phi = \{ \theta \phi : \theta \in \Theta, \phi \in \Phi \}.
\]
\end{definition}

\begin{definition} \label{D:DYADIC_CLASS_PRODUCT}
Let $\Theta$ and $\Phi$ be disjoint classes. Disjointness (definition \ref{D:DISJOINT_ENSEMBLES}) is implicitly implied by 
using dyadic notation: 
\[
\Theta \cup \Phi = \Theta \Phi.
\]
\end{definition}

\subsection{Deterministic-stochastic partition of reactive stimulus} \label{S:DETERMINISTIC_STOCHASTIC_PRTN} 
Definition \ref{D:REACTIVE_BASIS} asserts that basis $\Basis{\Phi}{\Xi}$ is comprised of two disjoint classes. 
Indeed $\Phi$, $\Xi$, and $\Psi = \Phi\Xi$ have a central role in systems theory, and this importance will now be discussed.

In systems theory, $\Phi$ represents read-write or \emph{deterministic} variables, and the class represents ensembles that can be 
written. Its name is due to the fact that once it is written, its value persists until overwritten, possibly many read operations 
later. The class $\Xi$ represents \emph{stochastic} variables, which are read-only. Thus there is a physical reason for the disjointness 
of these two classes: memory is either read-write or read-only, and not both.

The class $\Psi = \Phi\Xi$ represents the total \emph{reactive} stimulus (that is, the total input which uniquely determines 
the system's output or response). The system's response, on the other hand, must be recorded to read-write variables, which is 
$\Phi$ in the case. The mathematical functionality which transforms stimulus into response must have the functional prototype 
$\Phi\Xi \to \Phi$.

\section{Set product} \label{S:SET_PRODUCT} 
Informally, a choice space contains all combinations of variables' values, 
regardless of whether a particular combination is realized in a listing. The 
Cartesian (set) product formalizes this notion.

\subsection{Empirical range abstraction} \label{S:RANGE_ABSTRACTION} 
In science and engineering, ranges are empirical and not truly limit quantities -- they must be idealized to admit the possibility 
of values which are omitted or skipped. A value's absence from a particular data set does not imply that the value is impossible; 
thus the entire data set must undergo an operation called abstraction. For example, a collection of point set values may be 
abstracted as a bounded interval, whereas a categorical range such as $\{\mathrm{off, on}\}$ is not affected. 
This is properly an engineering operation, so it may be represented but not defined as a mathematical transformation.

We use the following pseudo-function to represent range abstraction:
\[
\mathrm{ideal\ range} = \Abstract{\mathrm{empirical\ range}}.
\]

\subsection{Closure} \label{S:CLOSURE} 
\emph{(Future: simplify by re-introducing appendix \ref{S:ABSTRACT_HULL} ff. as a closure.)}

\subsection{Abstract hull} \label{S:ABSTRACT_HULL} 
In science and engineering, some sets of stimuli are so numerous that they are 
not listable. A notation is needed to admit the possibility of values which are 
omitted or skipped in a given list of observations. 

An abstract hull is a finite subset whose total measure is T and contains no 
null subset.

A value's absence from a particular data set does not imply that the value is impossible; 
thus the entire data set must undergo an operation called abstraction. For example, a collection of point set values may be 
abstracted as a bounded interval, whereas a categorical range such as $\{\mathrm{off, on}\}$ is not affected. 
This is properly an engineering operation, so it may be represented but not defined as a mathematical transformation.

We use the following pseudo-function to represent range abstraction:
\[
\mathrm{ideal\ range} = \Abstract{\mathrm{empirical\ range}}.
\]

Empirical ranges will occur in the margins of a class:
\begin{definition} \label{D:MARGINAL_RANGE}
Let $\Psi$ be a class with $i \in \domain{\Psi}$. The \emph{marginal} range of $\Psi$ with respect to index $i$ is 
$\range{\Psi(i)} = \{ \psi(i) : \psi \in \Psi \}$.
\end{definition}

\begin{theorem} \label{T:ABSTRACTION_NEIGHBORHOOD}
Let $\Psi$ be a class. For each $\psi \in \Psi$ and $i \in \domain{\psi}$, $\psi(i) \in \Abstract{\range{\Psi(i)}}$. 
\end{theorem}
\begin{proof}
Abstraction includes some neighborhood of each observation. 
\end{proof}

We also assume that abstraction of the same phenomenon is invariant to the data set in which it occurs. 

\subsection{Cartesian product of empirical ranges} \label{S:GCP} 
\begin{definition} \label{D:MARGINAL_PROTOSETS}
Let $\Psi$ be a class. For each $i \in \domain{\Psi}$, set 
\[
\Psi_\heartsuit(i) = \Abstract{\range{\Psi(i)}}.
\]
\end{definition}

\begin{definition} \label{D:PROTOSET}
Let $\Psi$ be a class. The \emph{proto-set} $\Psi_\heartsuit$ is the union of the rationalized marginal ranges: 
\[
\Psi_\heartsuit = \bigcup_{i \in \domain{\Psi}} \Psi_\heartsuit(i).
\]
\end{definition}

\begin{definition} \label{D:PROTOSPACE}
Let $\Psi$ be a class with proto-set $\Psi_\heartsuit$. The \emph{proto-space} of $\Psi$ is the set ${\Psi_\heartsuit}^{\domain{\Psi}}$,
where  ${\Psi_\heartsuit}^{\domain{\Psi}}$ is the set of all mappings $\domain{\Psi} \to \Psi_\heartsuit$.
\end{definition}

\begin{definition} \label{D:CHOICE}
Let $\Psi$ be a class with proto-set $\Psi_\heartsuit$. Let $\chi \colon \domain{\Psi} \to \Psi_\heartsuit$ be a mapping.
If $\chi$ satisfies $\chi(i) \in \Psi_\heartsuit(i)$ for each $i \in \domain{\Psi}$, then $\chi$ is a \emph{choice} mapping of $\Psi$.
\end{definition}

\begin{definition} \label{D:CHOICE_SPACE}
The set of all choice mappings of class $\Psi$ is the \emph{choice} space (or general Cartesian product) $\prod\Psi$.
\end{definition}

\begin{remark}
When restricted to classes and ensembles, the general Cartesian product may be approximated as 
$\Closure{\Psi} = \{ \psi \in \Psi : \{ (i,v) : i \in \domain{\psi}, v \in \range{\Psi(i)} \} \}$.
\end{remark}

\begin{nomenclature}
Through the general Cartesian product, a class \emph{generates} a choice space. For brevity we refer to a point in a choice
space (that is, a choice mapping) simply as a \emph{choice}.
\end{nomenclature}

\section{Relations in choice space } \label{S:CHOICE_RELATIONS} 
\subsection{Proto-space} \label{S:PROTOSPACE} 
\begin{theorem}\label{T:PROTOSPACE_INCLUDES_CHOICESPACE}
The proto-space ${\Psi_\heartsuit}^{\domain{\Psi}}$ of class $\Psi$ includes its choice space $\prod\Psi$ (that is,
$\prod\Psi \subseteq {\Psi_\heartsuit}^{\domain{\Psi}}$).
\end{theorem}
\begin{proof}
Suppose $\chi \in \prod\Psi$. By definition \ref{D:CHOICE}, $\chi$ is a mapping $\domain{\Psi} \to \Psi_\heartsuit$. Then, by
definition \ref{D:PROTOSPACE}, $\chi \in {\Psi_\heartsuit}^{\domain{\Psi}}$. Thus, any member of $\prod\Psi$ is also a member of
${\Psi_\heartsuit}^{\domain{\Psi}}$. From this we conclude $\prod\Psi \subseteq {\Psi_\heartsuit}^{\domain{\Psi}}$.
\end{proof}

\subsection{Choice subspaces} \label{S:CHOICE_SUBSPACE} 
Any mapping, including a choice mapping, may be restricted to subsets of its domain.

\begin{definition} \label{D:SUBCHOICE}
Let $\Psi$ be an class, and let $R \subseteq \domain{\Psi}$ be a subset of its index set. Suppose $\chi \in \prod\Psi$ is a choice.
A \emph{subchoice} $\Restrict{\chi}{R}$ is the ordinary mapping restriction of $\chi$ to its domain subset $R$.
\end{definition}

\begin{remark}
In the above, degenerate case $R = \varnothing$ yields $\Restrict{\chi}{R} = \varnothing$.
\end{remark}

\begin{definition} \label{D:SUBSPACE}
Let $\Psi$ be an class. For each $R \subseteq \domain{\Psi}$, the \emph{subspace} $\Restrict{(\thinspace\prod\Psi)}{R}$ is the set
of subchoices $\Restrict{\lbrace\thinspace\nu}{R} \medspace \colon \nu \in \prod\Psi\thinspace\rbrace$.
\end{definition}

\begin{theorem}\label{T:CHC_RSTR_EQ_RSTR_CHC}
Let class $\Psi$ generate choice space $\prod \Psi$, and let $R \subseteq \domain{\Psi}$ be a subset of its index set.
The restriction of the choice space equals the choice space of the restriction:
\[
\Restrict{(\thinspace\prod\Psi)}{R} = \prod (\Restrict{\Psi}{R}).
\]
\end{theorem}
\begin{proof}
Suppose $\xi \in \Restrict{(\thinspace\prod\Psi)}{R}$. By definition \ref{D:SUBCHOICE}, there exists $\chi \in \prod\Psi$ such that
$\xi = \Restrict{\chi}{R}$. By definition of Cartesian product, for each $i \in \domain{\Psi}$, $\chi(i) \in \Psi(i)$. Since
$R \subseteq \domain{\Psi}$, then for each $r \in R$, $\xi(r) \in \Psi(r)$. Consider $\Restrict{\Psi}{R}$, for which
$\domain{(\Restrict{\Psi}{R})} = R$. By definition of restriction, for $r \in R$, $(\Restrict{\Psi}{R})(r) = \Psi(r)$. Since
$\xi(r) \in \Psi(r)$ and $\Psi(r) = (\Restrict{\Psi}{R})(r)$, then for any $r \in R$, $\xi(r) \in (\Restrict{\Psi}{R})(r)$~-- that is, $\xi$
is a choice of $\Restrict{\Psi}{R}$. From the preceding, $\xi \in \Restrict{(\thinspace\prod\Psi)}{R}$ implies 
$\xi \in \prod (\Restrict{\Psi}{R})$, or $\Restrict{(\thinspace\prod\Psi)}{R} \subseteq \prod (\Restrict{\Psi}{R})$.

Next suppose $\xi \in \prod (\Restrict{\Psi}{R})$. Then, by definitions \ref{D:CHOICE} and \ref{D:CHOICE_SPACE} covering Cartesian
products, for each $r \in R$, $\xi(r) \in (\Restrict{\Psi}{R})(r)$. The class $\Psi$ coincides with its restriction $\Restrict{\Psi}{R}$
on $R$. A restatement of this is $(\Restrict{\Psi}{R})(r) = \Psi(r)$ for $r \in R$. Substituting $\Psi(r)$ for $(\Restrict{\Psi}{R})(r)$
yields $\xi(r) \in \Psi(r)$ for each $r \in R$. From this it follows that $\xi \in \Restrict{(\thinspace\prod\Psi)}{R}$, with the further
implication that $\prod (\Restrict{\Psi}{R}) \subseteq \Restrict{(\thinspace\prod\Psi)}{R}$.

We conclude equality $\Restrict{(\thinspace\prod\Psi)}{R} = \prod (\Restrict{\Psi}{R})$ after establishing that each of these two sets is a
subset of the other.
\end{proof}

\begin{lemma}\label{L:SUBSPACE_SUBSET}
Let $\Psi$ and $\Phi$ be classes. If $\Closure{\Phi}$ is a subspace of $\Closure{\Psi}$, then $\Phi \subseteq \Psi$.
\end{lemma}
\begin{proof}
Let $\Closure{\Phi}$ be a subspace of $\Closure{\Psi}$. By definition \ref{D:SUBSPACE}, there exists $R \subseteq \domain{\Psi}$
such that $\Closure{\Phi} = \Restrict{(\thinspace\Closure{\Psi})}{R}$. By Theorem \ref{T:CHC_RSTR_EQ_RSTR_CHC}, 
the restriction of the choice space equals the choice space of the restriction:
$\Restrict{(\thinspace\prod\Psi)}{R} = \prod (\Restrict{\Psi}{R})$. Transitivity of equality implies $\Closure{\Phi} = \prod (\Restrict{\Psi}{R})$.
Then, by Theorem \ref{T:SPACE_UNIQ_ENSEMBLE} (invertibility of the Cartesian product), $\Phi = \Restrict{\Psi}{R}$.

Suppose term $(i,P) \in \Phi$. Since $\Phi = \Restrict{\Psi}{R}$, then $(i,P) \in \Restrict{\Psi}{R}$. By the definition of restriction,
this implies both $(i,P) \in \Psi$ and $i \in R$. Since $(i,P) \in \Phi$ implies $(i,P) \in \Psi$, we conclude $\Phi \subseteq \Psi$.
\end{proof}

\begin{lemma}\label{L:SUBSET_RESTRICTION}
Let $\Psi$ and $\Phi$ be classes. If $\Phi \subseteq \Psi$ and $R = \domain{\Phi}$, then $\Phi = \Restrict{\Psi}{R}$.
\end{lemma}
\begin{proof}
Consider $(i,P) \in \Restrict{\Psi}{R}$. It then follows from the definition of restriction that $(i,P) \in \Psi$ and $i \in R$. But
$R = \domain{\Phi}$, so $i \in \domain{\Phi}$. This implies there exists $(i,Q) \in \Phi$. Since $\Phi \subseteq \Psi$, then
$(i,Q) \in \Psi$. Since $\Psi$ is a mapping, then $(i,P) \in \Psi$ and $(i,Q) \in \Psi$ implies $P = Q$. From $P = Q$ and
$(i,Q) \in \Phi$, we infer that $(i,P) \in \Phi$. Thus $(i,P) \in \Restrict{\Psi}{R}$ implies $(i,P) \in \Phi$, so
$\Restrict{\Psi}{R} \subseteq \Phi$.

Next suppose $(i,P) \in \Phi$. From this it follows that $i \in R = \domain{\Phi}$. From the premises $(i,P) \in \Phi$ and
$\Phi \subseteq \Psi$ we conclude $(i,P) \in \Psi$. Together $(i,P) \in \Psi$ and $i \in R$ imply that $(i,P) \in \Restrict{\Psi}{R}$.
Thus $(i,P) \in \Phi$ implies $(i,P) \in \Restrict{\Psi}{R}$, so $\Phi \subseteq \Restrict{\Psi}{R}$.

From $\Restrict{\Psi}{R} \subseteq \Phi$ and $\Phi \subseteq \Restrict{\Psi}{R}$ we infer $\Phi = \Restrict{\Psi}{R}$.
\end{proof}

\begin{lemma}\label{L:SUBSET_SUBSPACE}
Let $\Psi$ and $\Phi$ be classes. If $\Phi \subseteq \Psi$, then $\Closure{\Phi}$ is a subspace of $\Closure{\Psi}$.
\end{lemma}
\begin{proof}
Set $R = \domain{\Phi}$. Since $\Phi \subseteq \Psi$ by hypothesis, then by applying lemma \ref{L:SUBSET_RESTRICTION} we infer
$\Phi = \Restrict{\Psi}{R}$. With this equality and  Theorem \ref{T:SPACE_UNIQ_ENSEMBLE} (invertibility of the Cartesian product),
we have $\Closure{\Phi} = \prod (\Restrict{\Psi}{R})$. Theorem \ref{T:CHC_RSTR_EQ_RSTR_CHC} asserts that the restriction of the choice space
equals the choice space of the restriction: $\Restrict{(\thinspace\prod\Psi)}{R} = \prod (\Restrict{\Psi}{R})$. Transitivity of equality implies
$\Closure{\Phi} = (\thinspace\prod\Psi) \mid R$. This last equality is exactly the premise of definition \ref{D:SUBSPACE}: $\Closure{\Phi}$
is a subspace of $\Closure{\Psi}$.
\end{proof}
\begin{theorem}\label{T:SUBSET_IFF_SUBSPACE}
Let $\Psi$ and $\Phi$ be classes. $\Closure{\Phi}$ is a subspace of $\Closure{\Psi}$ if and only if $\Phi \subseteq \Psi$.
\end{theorem}
\begin{proof}
lemma \ref{L:SUBSPACE_SUBSET} asserts that if $\Closure{\Phi}$ is a subspace of $\Closure{\Psi}$, then $\Phi \subseteq \Psi$.
Lemma \ref{L:SUBSET_SUBSPACE} asserts that if $\Phi \subseteq \Psi$, then $\Closure{\Phi}$ is a subspace of $\Closure{\Psi}$.
This pair of inverse implications establishes the biconditional.
\end{proof}

\begin{lemma}\label{L:ENSEMBLE_PROD_SUBSETS}
If $\Upsilon$, $\Psi$, and $\Phi$ are classes such that $\Upsilon = \Psi\Phi$, then $\Psi \subseteq \Upsilon$ and
$\Phi \subseteq \Upsilon$. 
\end{lemma}
\begin{proof}
Since $\Upsilon$ is the dyadic product of $\Psi$ and $\Phi$, then by definition \ref{D:DYADIC_ENSEMBLE_PRODUCT}, $\Psi$ and $\Phi$
are disjoint classes and $\Upsilon = \Psi \cup \Phi$.

Suppose $i \in \domain{\Upsilon} = \domain{(\Psi \cup \Phi)}$. Through definition \ref{D:DISJOINT_ENSEMBLES}, disjointness entails
that $\domain{\Psi}\thickspace\cap\thickspace\domain{\Phi} = \varnothing$. Thus, if $i \in \domain{\Upsilon}$, exactly one of two
cases hold: either A: $i \in \domain{\Psi}$ and $i \notin \domain{\Phi}$, or B: $i \in \domain{\Phi}$ and $i \notin \domain{\Psi}$.

Assume case A, that $i \in \domain{\Psi}$ and $i \notin \domain{\Phi}$. With $\Upsilon = \Psi \cup \Phi$, it follows from the
definition of set union that for any $i \in \domain{\Psi}$, $(i,P) \in \Psi$ implies $(i,P) \in \Upsilon$~-- that is,
$\Psi \subseteq \Upsilon$.

For case B, similar argument leads to $\Phi \subseteq \Upsilon$.
\end{proof}

\begin{corollary}\label{C:ENSEMBLE_PROD_MEMBERS}
If $\Upsilon$, $\Psi$, and $\Phi$ are classes such that $\Upsilon = \Psi\Phi$, then $\Psi(i) = \Upsilon(i)$ for
$i \in \domain{\Psi}$, and $\Phi(j) = \Upsilon(j)$ for $j \in \domain{\Phi}$.
\end{corollary}
\begin{proof}
Under identical premises, lemma \ref{L:ENSEMBLE_PROD_SUBSETS} provides $\Psi \subseteq \Upsilon$ and $\Phi \subseteq \Upsilon$.
Suppose $i \in \domain{\Psi}$. If $(i,P) \in \Psi$, then $(i,P) \in \Upsilon$ since $\Psi \subseteq \Upsilon$. The notation
$\Psi(i) = \Upsilon(i)$ (both equaling $P$) is equivalent. A similar argument demonstrates $\Phi(j) = \Upsilon(j)$ for
$j \in \domain{\Phi}$.
\end{proof}

\begin{theorem}\label{T:ENSEMBLE_PROD_CHOICE_PROD}
Let $\Upsilon$, $\Psi$, and $\Phi$ be classes such that $\Upsilon = \Psi\Phi$. For each $\upsilon \in \Closure{\Upsilon}$, there exist
unique $\psi \in \Closure{\Psi}$ and $\phi \in \Closure{\Psi}$ such that $\upsilon = \psi\phi$.
\end{theorem}
\begin{proof}
Suppose $\upsilon \in \Closure{\Upsilon}$. Since any choice has the same domain as its generating ensemble,
$\domain{\Upsilon} = \domain{\upsilon}$. Theorem \ref{T:DYADIC_PRODUCT_IS_ENSEMBLE} states that
$\domain{\Upsilon} = \domain{\Psi} \cup \domain{\Phi}$, from which transitivity of equality provides
$\domain{\upsilon} = \domain{\Psi} \cup \domain{\Phi}$.

From lemma \ref{L:ENSEMBLE_PROD_SUBSETS} we conclude $\Psi \subseteq \Upsilon$ and $\Phi \subseteq \Upsilon$. Since these relations
hold for entire classes, then the same is true of the ensembles' domains: $\domain{\Psi} \subseteq \domain{\Upsilon}$ and
$\domain{\Phi} \subseteq \domain{\Upsilon}$. By substitution, the previous result $\domain{\Upsilon} = \domain{\upsilon}$ then
establishes that $\domain{\Psi} \subseteq \domain{\upsilon}$ and $\domain{\Phi} \subseteq \domain{\upsilon}$.

The inclusion $\domain{\Psi} \subseteq \domain{\upsilon}$ ensures that the restriction
$\psi = \Restrict{\upsilon}{\domain{\Psi}}$ is well-defined. Similarly $\phi = \Restrict{\upsilon}{\domain{\Phi}}$ is also well-defined.

We next focus on the restriction $\psi$ constructed above, seeking to demonstrate that it is also a member of the choice space
$\Closure{\Psi}$. Suppose term $(i, p) \in \psi$. Since $\psi = \Restrict{\upsilon}{\domain{\Psi}}$, then both $(i, p) \in \upsilon$ and
$i \in \domain{\Psi}$. Since $\upsilon \in \Closure{\Upsilon}$ by hypothesis, definition \ref{D:CHOICE} demands that $p \in \Upsilon(i)$
whenever $(i, p) \in \upsilon$. Corollary \ref{C:ENSEMBLE_PROD_MEMBERS} asserts $\Psi(i) = \Upsilon(i)$ for $i \in \domain{\Psi}$.
Since $p \in \Upsilon(i)$ and $\Upsilon(i) = \Psi(i)$ then $p \in \Psi(i)$. Thus for any $(i,p) \in \psi$, it follows that
$p \in \Psi(i)$. This means that $\psi$ is a choice of $\Psi$ by definition \ref{D:CHOICE}~-- that is, $\psi \in \Closure{\Psi}$
by definition \ref{D:CHOICE_SPACE}.

Similar reasoning establishes that the other restriction $\phi$ is a member of $\Closure{\Phi}$. The unique $\psi \in \Closure{\Psi}$ and
$\phi \in \Closure{\Psi}$ such that $\upsilon = \psi\phi$ are expressed by the restrictions $\psi = \Restrict{\upsilon}{\domain{\Psi}}$ and
$\phi = \Restrict{\upsilon}{\domain{\Phi}}$.
\end{proof}

\chapter{Invalid logic in MIL-STD-882} \label{S:MIL-STD-882-LOGIC} 
A secondary goal is resolving the conflict over the definition of the term risk. 
The first meaning is an accident model known in mathematics as the compound Poisson stochastic process, and is called \emph{statistical} risk. 
It is composed of an (intensity, severity) pair, and underlies the theory of safety and availability engineering. 
Reliability considers only the dimension of intensity. 
The second is here called \emph{developmental} risk, and considers a sequence of five categories of decreasing software autonomy, which allow decreasing amounts of time for manual recovery of control.

A secondary topic is distinguishing two separate subjects which have become confused in \Acro{MIL-STD-882E}: there are two different technical meanings for the term \emph{risk}. This lack of clarity involves an engineered product assurance method and an engineering process assurance method for software development.

The first meaning is known in mathematics as the compound Poisson stochastic process, and is called \emph{statistical} risk. 
It is composed of an (intensity, severity) pair, and underlies the theory of safety and availability engineering. 
Reliability considers only the dimension of intensity. 
The second is here called \emph{developmental} risk, and considers a sequence of five categories of decreasing software autonomy, which allow decreasing amounts of time for manual recovery of control.

A primary objective of safety engineering is consistent quantitative risk assessment. Consistency was enforced through a common risk model known in 
mathematics as the compound Poisson stochastic process. 

A primary objective of safety engineering is consistent quantitative 
(or as an approximation, categorical) risk assessment. Consistency was enforced through a common risk model known in 
mathematics as the compound Poisson stochastic process.

Such units characterize military standard safety and reliability measurements, particularly \Acro{MIL-STD-882E} (the United States Department 
of Defense standard for system safety engineering).

A second incentive for this memoir is correcting an error in 
\Acro{MIL-STD-882E}.  
This consistency was recently lost 
to a new empirical method for software risk management. In the current 
Standard, hardware risk management is conducted via the compound Poisson process, while 
software is assigned the empirical method. Comparison of these methods reveals potentially 
noxious divergence. This conclusion is inescapable because the empirical method is not 
algebraically reducible to a compound Poisson process, and fails to be a measure of risk 
according to dimensional analysis. At most one of these methods can be correct. The empirical 
method is presently pseudo-science. To gain status as proper science. its scope must be examined for relevance to elementary engineering principles, \emph{\Acro{NOT}} risk assessment.

A second incentive for this memoir is correction of an inconsistency in \Acro{MIL-STD-882E}. This inconsistency involves a regimin that begins with a general product assurance method and concludes confounded with a process assurance method for software development. This confusion's enabling mechanism is two different technical meanings for the term \emph{risk}.
The first meaning is known in mathematics as the compound Poisson stochastic process, and will be called \emph{statistical} risk. 
For lack of a better term, the second is called \emph{developmental} risk.
In development risk, there is but one universal hazard: Loss of computer control. It is mitigated by the the duration required for a human operator to recognize and manually recover from loss of computer control.

Computers may control safety-critical operations in machines having embedded software. This memoir proposes a regimen to verify such algorithms at precribed 
levels of statistical confidence. 

Discussion appears in two major parts: theory, which shows the 
relationship between automata, discrete systems, and software; 
and application, which covers demonstration and indemnification. 
Demonstration is a method for generating random tests and indemnification 
is a technique for representing pass/fail test results as compound Poisson 
measure (severity and intensity). 

A second incentive for this memoir is correction of an inconsistency in \Acro{MIL-STD-882E}. This inconsistency involves a regimin that begins with a general product assurance method and concludes confounded with a process assurance method for software development. This confusion's enabling mechanism is two different meanings for the technical term \emph{risk}.
The first meaning is known in mathematics as the compound Poisson stochastic process, and will be called \emph{statistical} risk. 
The second is called \emph{developmental} risk.

A primary objective of safety engineering is consistent quantitative risk assessment. Consistency was enforced through a common risk model known in 
mathematics as the compound Poisson stochastic process. 

A primary objective of safety engineering is consistent quantitative 
(or as an approximation, categorical) risk assessment. Consistency was enforced through a common risk model known in 
mathematics as the compound Poisson stochastic process.

Such units characterize of military standard safety and reliability measurements, particularly \Acro{MIL-STD-882E} (the United States Department 
of Defense standard for system safety engineering).

A second incentive for this memoir is correcting an error in 
\Acro{MIL-STD-882E}.  
This consistency was recently lost 
to a new empirical method for software risk management. In the current 
Standard, hardware risk management is conducted via the compound Poisson process, while 
software is assigned the empirical method. Comparison of these methods reveals potentially 
noxious divergence. This conclusion is inescapable because the empirical method is not 
algebraically reducible to a compound Poisson process, and fails to be a measure of risk 
according to dimensional analysis. At most one of these methods can be correct. The empirical 
method is presently pseudo-science. To gain status as proper science. its scope must be examined for relevance to elementary engineering principles, \emph{\Acro{NOT}} risk assessment.

A second incentive for this memoir is correction of an inconsistency in \Acro{MIL-STD-882E}. This inconsistency involves a regimin that begins with a general product assurance method and concludes confounded with a process assurance method for software development. This confusion's enabling mechanism is two different technical meanings for the term \emph{risk}.
The first meaning is known in mathematics as the compound Poisson stochastic process, and will be called \emph{statistical} risk. 
The second is called \emph{developmental} risk for lack of a better term.

\chapter{MIL-STD-882 and the categorical CPP} \label{S:MIL-STD-882} 

\MilStd{} is the United States Department of Defense Standard Practice for System Safety. Revision~E became effective 
May 11, 2012. In preference to \emph{accident}, this standard prefers the term \emph{mishap}, which it defines as \Dquo{an event or 
series of events resulting in unintentional death, injury, occupational illness, damage to or loss of equipment or property, or damage 
to the environment.}

Its safety risk assessment method uses the compound Poisson process (CPP) to represent the timing and severity
of mishaps. \MilStdE{} partitions compound Poisson processes into a lattice of categories and levels that covers the range of interest.
The \emph{category} is a variable which, in an explicit range [1-4], expresses the expectation $(\mu_L)$ of the CPP loss 
random variable $L$.
The \emph{level}  is a variable which, in an explicit range [A-F], expresses the rate or intensity $\lambda$ of the CPP.

The system of categories and levels agrees with the limits of discernibility of human intuition. Two different compound Poisson 
processes having the same category and level are indeed different but in practice are indistinguishable. This characteristic imposes 
a logarithmic organization on the categories and levels.  

\begin{table}[H]
    \begin{center}
        \begin{tabular}{|l|c|p{4in}|}
            \hline
Description& Severity&Mishap Result Criteria \\
           &Category&                        \\ \hline
Catastrophic&1&Could result in one or more of the following: death, permanent total disability, irreversible significant environmental impact, or monetary loss equal to or exceeding \$10M. \\ \hline
Critical&2&Could result in one or more of the following: permanent partial disability, injuries or occupational illness that may result in hospitalization of at least three personnel, reversible significant environmental impact, or monetary loss equal to or exceeding \$1M but less than \$10M. \\ \hline
Marginal&3&Could result in one or more of the following: injury or occupational illness resulting in one or more lost work day(s), reversible moderate environmental impact, or monetary loss equal to or exceeding \$100K but less than \$1M. \\ \hline
Negligible&4&Could result in one or more of the following: injury or occupational illness not resulting in a lost work day, minimal environmental impact, or monetary loss less than \$100K. \\
            \hline
        \end{tabular}
    \end{center}
    \caption{MIL-STD-882E Severity Categories} \label{Ta:SEVERITY_CATEGORY}
\end{table}

\begin{table}[H]
    \begin{center}
        \begin{tabular}{|l|c|p{2in}|p{2in}|}
            \hline
Description&Level&Specific Individual Item&Fleet or Inventory \\ \hline
Frequent&A&Likely to occur often in the life of an item.&Continuously experienced. \\ \hline
Probable&B&Likely to occur often in the life of an item.&Will occur frequently. \\ \hline
Occasional&C&Likely to occur sometime in the life of an item.&Will occur several times. \\ \hline
Remote&D&Unlikely, but possible to occur in the life of an item.&Unlikely, but can reasonably be expected to occur. \\ \hline
Improbable&E&So unlikely, it can be assumed occurrence may not be experienced in the life of an item.&Unlikely to occur, but possible. 
\\ \hline
Eliminated&F&\multicolumn{2}{|p{4in}|}{Incapable of occurrence. This level is used when potential hazards are identified and 
later eliminated.} \\ \hline
        \end{tabular}
    \end{center}
    \caption{MIL-STD-882E Probability Levels} \label{Ta:PROBABILITY_LEVELS}
\end{table}

Table 2 above is a qualitative description of levels. Table 3 below, appearing in \MilStdE{} Appendix A, outlines certain pitfalls 
in accomplishing the same task quantitatively. Numerical expression of the intensity or rate of occurrence is generally preferable to 
mere qualitative phrasing. For quantitative description, the intensity is the ratio of mishaps (numerator) to some measure of exposure 
(denominator). 

\begin{table}[H]
    \begin{center}
        \begin{tabular}{|l|c|p{1.5in}|p{1.5in}|p{1.5in}|}
            \hline
Description&Level&Individual Item&Fleet/Inventory*&Quantitative \\ \hline
Frequent&A&Likely to occur often in the life of an item&Continuously experienced.&Probability of occurrence greater than or equal to 
$10^{-1}$. \\ \hline
Probable&B&Will occur several times in the life of an item&Will occur frequently.&Probability of occurrence less than $10^{-1}$ but greater than or equal to $10^{-2}$. \\ \hline
Occasional&C&Likely to occur sometime in the life of an item&Will occur several times.&Probability of occurrence less than $10^{-2}$ but greater than or equal to $10^{-3}$. \\ \hline
Remote&D&Unlikely, but possible to occur in the life of an item&Unlikely but can reasonably be expected to occur.&Probability of occurrence less than $10^{-3}$ but greater than or equal to $10^{-6}$. \\ \hline
Improbable&E&So unlikely, it can be assumed occurrence may not be experienced in the life of an item&Unlikely to occur, but possible.&Probability of occurrence less than $10^{-6}$. \\ \hline
Eliminated&F&\multicolumn{3}{|p{4.5in}|}{Incapable of occurrence within the life of an item. This category is used when potential hazards are identified and later eliminated.} \\ \hline
        \end{tabular}
        \caption{MIL-STD-882E Example Probability Levels} 
             {\footnotesize{* The size of the fleet or inventory should be defined}} 
             \label{Ta:EXAMPLE_PROBABILITY_LEVELS}
    \end{center}
\end{table}

The false hegemony of a single intuitively understood measure of exposure will now be examined. 
We will find that, however well-intended, Table 3 lacks essential explanation. Without that, it is an oversimplification.

\Dquo{Natural} measures of exposure must embrace a variety of units, some examples of which are: the life of one item, number 
of missile firings, flight hours, miles driven, or years of service. For example, an exposure measure of miles driven is expected 
silently to exclude substantial periods when the system is out of use. Similar would be any situation-based measure of exposure 
having a sizable portion of time spent in unused status (time not counted). This topic appeared in \S\ref{S:INTERMITTENT_POISSON},
the intermittant compound Poisson process.
The natural unit of exposure can be tuned to the culture of a particular hazard. 
However, lacking conversion capability, this freedom of choice leads to the problem of a system composed of a 
heterogeneous plethora of non-comparable exposure units. 

What is behind this incomparability? Natural units are important but incomplete --  
\MilStdE{} needs additional factors to paint a full quantitative picture. 
There is need for conversion of various natural units into a single common standard unit, so that comparison involves only observation of
magnitudes, without pondering the meaning of different units. 
This is particularly important in the cases of many ambiguous references to \Dquo{life.}
Suppose we arbitrarily standardize time duration at one year.
We then define a conversion factor $p$, which means that $p$ years constitute a life. 
A measure $\iota$ quantifies what fraction of time the system's mission is inactive or idle. 
A conversion factor for remaining units must be established; 
without specifying what units remain to be converted, we can say that the unit conversion calculus of elementary physics results 
in some linear coefficient $\kappa$. With $N$ a natural exposure unit and $U$ a standard measure, what we have stated so far is 
summarized in the following form: 
\[
U = \frac{\kappa \cdot (1 - \iota)}{p} \cdot N .
\]

Standard units measure statistical risk as resulting from exposure to an intermittent compound Poisson process.
These standard units may not be a proper exposure, but measure the exposure expected in a year's duration.
For this reason we celebrate the importance of the role of pure natural units; it is important to understand risk as proportionate
to exposure. To understand this importance, imagine yourself as the one exposed to a transient but intense hazard.
But that does not imply the dismissal of statistical risk as a concern; it is also part of the risk analysis picture
to consider how much risk exposure occurs within a given duration. This is the role of the standard unit.

Another complicating factor is the use of the term \Dquo{level} itself. A level is a designator for a class of possibly intermittent 
indistinguishable probability distributions. Rather than being clear about this, \MilStdE{} equivocates greatly in Table 3, confusing 
this designator with a literal probability statement. Only after full quantitative analysis is completed ($p$, $\iota$, and $\kappa$ 
known) can definite statements concerning probability be asserted. It is insufficient to mandate vague documentation of \Dquo{all numerical 
definitions of probability used in risk assessments} without further guidance.

Table 4 below is a categorical rendering of the hyperbola of statistical risk. Definition \ref{D:STATISTICAL_RISK} asserts
$h = \lambda \mu_L$. Excepting the administrative level \Dquo{Eliminated}, this cross-tabulation presents the level $(\lambda)$ 
on the vertical axis and the category $(\mu_L)$ along the horizontal axis. For each combination of level and category, another 
categorical variable\footnote{Not to be confused with the categorical variable named \Dquo{category}} 
represents the statistical risk $h = \lambda \mu_L$ with values: High, Serious, Medium, and Low.

\begin{table}[h!]
    \begin{center}
        \begin{tabular}{|l|c|c|c|c|}
            \hline
SEVERITY /&Catastrophic&Critical&Marginal&Negligible \\
PROBABILITY&(1)&(2)&(3)&(4) \\ \hline
Frequent (A)&High&High&Serious&Medium \\ \hline
Probable (B)&High&High&Serious&Medium \\ \hline
Occasional (C)&High&Serious&Medium&Low \\ \hline
Remote (D)&Serious&Medium&Medium&Low \\ \hline
Improbable (E)&Medium&Medium&Medium&Low \\ \hline
Eliminated (F)&\multicolumn{4}{|c|}{Eliminated} \\ \hline
        \end{tabular}
    \end{center}
    \caption{MIL-STD-882E Risk Assessment Matrix} \label{Ta:RISK_ASSESSMENT_MATRIX}
\end{table}

This table suffers the same ambiguity as in Table 3. \MilStd's definitions are clearly inadequate for quantitative analysis. 
Through equivocation, exposure to an intermittent compound Poisson process is regarded as not different than exposure to a compound 
Poisson process, despite that the difference becomes obvious through the linear factor $(1 - \iota)$. \MilStd{} is an evolving 
document in its fifth major revision; let us hope these ambiguities are 
resolved in the future.

\chapter{Deterministic finite automaton} 
While neither the deterministic finite automaton nor the reactive actuated automaton are proper subsets of each other, 
they do have similarities.

\section{Definition of DFA} \label{S:DFA} 
This simple depiction of the deterministic finite automaton appears in Wikipedia \cite{wW11autmaton}:
\begin{quotation}
\noindent A [deterministic finite] \emph{automaton} is represented formally by the 5-tuple $\langle Q, \Sigma, \delta, q_0, A \rangle$, where:
\begin{itemize}
  \item $Q$ is a finite set of \emph{states}.
  \item $\Sigma$ is a finite set of \emph{symbols}, called the \emph{alphabet} of the automaton.
  \item $\delta$ is the \emph{transition function}, that is, $\delta \colon Q \times \Sigma \to Q$.
  \item $q_0$ is the \emph{start state}, that is, the state which the automaton \emph{occupies} when no input has been processed yet,
        where $q_0 \in Q$.
  \item $A$ is a set of states of $Q$ (i.e. $A \subseteq Q$) called \emph{accept states}.
\end{itemize}
\end{quotation}

An approach for engineers is found in \cite{jH79}.

\section{Analogous structures} \label{S:ANALOGOUS_DFA_RAA} 
The following structures are analogous between the DFA and RAA:

\begin{table}[h!]
    \begin{center}
        \begin{tabular}{|l||c|c|}
            \hline
Analogy&DFA&RAA \\ \hhline{=#=#=}
State/Locus&Q&$\Lambda$ \\ \hline
Alphabet/Reactive Stimulus&$\Sigma$&$\Closure{\Psi}$ \\ \hline
Transition Function/Actuation&$\delta \colon Q \times \Sigma \to Q$&
        $\lambda_{i+1} = \tranaddr{\Act{a}_i}{\psi_i}$ \\ \hline
        \end{tabular}
    \end{center}
    \caption{DFA/RAA Analogies} \label{Ta:DFA_RAA_ANALOGIES}
\end{table}

\section{Actuation vs. transition function} \label{S:TRANSITION_FUNCTION} 
The reactive actuated automaton uses actuators to determine the next locus, while the deterministic finite 
automaton uses the transition function for the same purpose. In the DFA, if $q \in Q$ is the current state and 
$\sigma \in \Sigma$ is the current symbol, then the next state is $q' = \delta(q, \sigma)$. For the RAA, 
assuming $\lambda \in \Lambda$ is the current locus and $\psi \in \Psi$ is the current reactive stimulus, then 
the next locus is $\lambda' = \tranaddr{\Act{a}_i}{\psi_i}$.

The RAA is somewhat more restrictive because both its next locus and its actuated functionality must remain constant 
over the blocks of partition $\#_\lambda(\Closure{\Psi})$.
\begin{remark}
This means that the actuated functionality $\Ftylc{f}$ is the same across partition block $B$. 
It does not that the application 
$\Ftylc{f}(\psi)$ is constant over $\psi \in B$.
\end{remark}

\chapter{Locus digraph} \label{S:LOCUS_DIGRAPH} 
With any reactive actuated automaton is associated a directed graph called the locus digraph, also known as the software control flow graph. 
This digraph is made famous by Thomas McCabe, who uses it in his study of program (cyclomatic) complexity\cite{wW16cyclomatic}.

\section{Vertex} \label{S:LOCUS_DIGRAPH_VERTEX} 
\begin{definition}
Let $\Auto{A}$ be a reactive actuated automaton. 
In the locus digraph induced by $\Auto{A}$, a vertex is an element of the automaton's catalog of loci $\Lambda$.
\end{definition}

\section{Edge} \label{S:LOCUS_DIGRAPH_EDGE} 
\begin{definition}
In the locus digraph induced by automaton $\Auto{A}$, a vertex $\lambda$ is directionally connected to another vertex $\lambda'$ if there exists 
configuration $\ConfigScript{c}$ such that $\lambda = \mho_\Lambda(\ConfigScript{c})$ and $\lambda' = \mho_\Lambda(\Auto{A}(\ConfigScript{c}))$.
\end{definition}


\pagebreak

\end{document}